\documentclass[runningheads,envcountsect,envcountsame]{llncs}

\usepackage{ifthen}
\usepackage[utf8]{inputenc}
\usepackage{amssymb}
\usepackage{amsmath}
\usepackage{tikz}
\usepackage{paralist}
\usepackage{thm-restate}

\newcommand{\paramBox}[1]{\ensuremath{\Box_{#1}}}

\newcommand{\lmplus}[1]{\ensuremath{\mathsf{ML}^{#1}}}
\newcommand{\lmA}[1]{\ensuremath{\mathsf{ML}^{A(#1)}}}
\newcommand{\lmAp}[1]{\ensuremath{\mathsf{ML}^{A(#1)^+}}}
\newcommand{\node}[2]{\ensuremath{\langle\mathbb #1\circ\mathbb #2\rangle}}
\newcommand{\angNode}[2]{\ensuremath{\langle#1\circ#2\rangle}}
\newcommand{\boxlabels}[1]{\ensuremath{\mathop{labels}_\Box(#1)}}
\newcommand{\depth}[1]{\ensuremath{\mathop{depth}(#1)}}
\newcommand{\md}[1]{\ensuremath{\mathop{md}(#1)}}

\newcommand{\dcard}[1]{\ensuremath{\card{\card{#1}}}}
\newcommand{\altword}[1]{\variableLAltword{\ell}{#1}}
\newcommand{\variableLAltword}[2]{\ensuremath{a^{#1}_{#2}}}
\newcommand{\basemodel}[3]{\ensuremath{\mathbb{#1}^{#2}_{#3}}}
\newcommand{\starmodel}[3]{\ensuremath{\mathbb{#1}^{#2,*}_{#3}}}
\newcommand{\baseA}{\basemodel{A}{\ell}{i}}
\newcommand{\baseB}{\basemodel{B}{\ell}{s}}

\newcommand{\starA}{\starmodel{A}{\ell}{i}}
\newcommand{\starB}{\starmodel{B}{\ell}{s}}

\newcommand{\wtrap}{\ensuremath{w_{\mathtext{trap}}}}
\newcommand{\modelClass}[3]{\ensuremath{\mathbb{#1}^{#2,*}_{#3}}}
\newcommand{\modelClassA}{\modelClass{A}{\ell}{i}}
\newcommand{\modelClassB}{\modelClass{B}{\ell}{i}}
\newcommand{\oneModelClassA}{\modelClass{A}{\ell}{1}}
\newcommand{\oneModelClassB}{\modelClass{B}{\ell}{1}}

\newcommand{\effRewrite}{\ensuremath{\leq_{poly}}}
\newcommand{\leqExpress}{\ensuremath{\leq_{expr}}}
\newcommand{\set}[1]{\ensuremath\left\{#1\right\}}
\newcommand{\mathtext}[1]{\ensuremath{\mathrm{\text{#1}}}}
\newcommand{\card}[1]{\left| #1 \right|}
\newcommand{\ssf}{\ensuremath{O}}
\newcommand{\setSsf}{\ensuremath{\mathcal O}}

\newcommand{\complexityclassname}[1]{\ensuremath{\mathrm{#1}}}
\newcommand{\PTIME}{\complexityclassname{P}}
\newcommand{\PSPACE}{\complexityclassname{PSPACE}}

\newcommand{\paramParseTree}[3]{\ensuremath{T^{#1}(\angNode{#2}{#3})}}
\newcommand{\parseTree}{\paramParseTree{\psi}{\modelClassA}{\modelClassB}}

\newcommand{\lmF}{\ensuremath{\lmplus{\mathcal F}}}
\newcommand{\lmG}{\ensuremath{\lmplus{\mathcal G}}}

\title{The Relative Succinctness and Expressiveness of Modal Logics Can Be Arbitrarily Complex}
\author{Henning Schnoor}
\institute{Institut f\"ur Informatik, Christian-Albrechts-Universit\"{a}t zu Kiel, 24098 Kiel, Germany \email{henning.schnoor@email.uni-kiel.de}}
\titlerunning{Relative Succinctness and Expressiveness Can Be Arbitrarily Complex}

\begin{document}

\maketitle

\setcounter{secnumdepth}{5}
\setcounter{tocdepth}{5}

\begin{abstract}
We study the relative succinctness and expressiveness of modal logics, 
and prove that these relationships can be as complex as any
countable partial order. For this, we use two uniform formalisms to define 
modal operators, and obtain results on succinctness and expressiveness in these
two settings. Our proofs are based on formula size games introduced by
Adler and Immerman and bisimulations.
\end{abstract}

\section*{Introduction}

Modal logics of different flavours play an important role in computer science, especially as specification languages (\cite{hamo92,FaginHalpernMosesVardi-REASONING-ABOUT-KNOWLEDGE-MITPRESS-1995,AucherBoellavdTorre-PRIVACY-WITH-MOCAL-LOGIC-DEON-2010,BlackburnDeRijkeVenama-MODAL-LOGIC-BOOK-2001}). Therefore, the study of expressiveness and succinctness of modal and other logics have received much attention: In~\cite{GogicKautzPapadimitriouSelman-COMPARATIVE-LINGUISTICS-KNOWLEDGE-REPRESENTATION-IJCAI-1995}, the succinctness of different formalisms to define knowledge bases was compared. In~\cite{Wilke-CTL-SUCCINCTNESS-FSTTCS-1999}, it was shown that CTL$^+$ is exponentially more succinct than CTL, i.e., in the translation from CTL$^*$ to CTL, an exponential blow-up in formula size cannot be avoided. This result was later strengthened in~\cite{AdlerImmerman-LOWER-BOUND-FORMULA-SIZE-TOCL-2003}. The techniques of the latter paper, \emph{formula size games}, were applied to modal logics in \cite{FrenchVanDerHoekIlievKooi-SUCCINCTNESS-MODAL-LOGIC-AI-2013} and \cite{VanderhoekIliev-RELATIVE-SUCCINCTNESS-AAMAS-2014-ACM-ENTRY}. They show that, in particular, there are modal logics $\mathcal L_1$ and $\mathcal L_2$ such that $\mathcal L_1$ is exponentially more succinct than $\mathcal L_2$ and vice versa. (This seemingly contradictory statement says that some properties are more efficiently expressed in $\mathcal L_1$, and some in $\mathcal L_2$).

This result raises several questions: Are there arbitrary large sets of modal logics, where each logic is exponentially more succinct than all of the others? Are there arbitrarily long sequences of modal logics of strictly increasing succinctness? More generally, can the ``succinctness''-relationships between modal logics be arbitrarily complex? 

Formally, let $\effRewrite$ be the relation between modal logics such that $\mathcal L_1\effRewrite\mathcal L_2$ if for every $\mathcal L_1$-formula, there is an equivalent $\mathcal L_2$-formula of polynomial size. The results from the above-mentioned~\cite{VanderhoekIliev-RELATIVE-SUCCINCTNESS-AAMAS-2014-ACM-ENTRY} imply that $\effRewrite$ is not a linear order, but clearly, $\effRewrite$ is reflexive and transitive. Does $\effRewrite$ have any other standard properties in addition to reflexivity and transitivity?

We answer the above questions by showing that $\effRewrite$ can be as complex as any countable partial order. More precisely, for any partial order $\leq_S$ on a countable set $S$, we exhibit a family of modal logics $(\mathcal L_s)_{s\in S}$, all equally expressive, such that $\effRewrite$ on $(\mathcal L_s)_{s\in S}$ behaves exactly like $\leq_S$ in the following sense: If $s_1\leq_Ss_2$, then $\mathcal L_{s_1}\effRewrite\mathcal L_{s_2}$ and $\mathcal L_{s_1}$ is exponentially more succinct than $\mathcal L_{s_2}$ otherwise. In particular, there is indeed an infinite set of modal logics where each logic is exponentially more succinct than every other, and there is an infinite sequence of modal logics, each strictly more succinct than the previous one. For the related question of expressiveness, we get analogous results: If $\leqExpress$ is defined as $\mathcal L_1\leqExpress\mathcal L_2$ if for every $\mathcal L_1$-formula, there is some equivalent $\mathcal L_2$-formula (regardless of the size), then $\leqExpress$ can be as complex as any countable partial order in exact same sense as above.

To prove our results, we use a uniform way to define modal logics. We consider two different ways to define generalized (multi-)modal operators:
 \begin{enumerate}
  \item ``One-Step'' modal operators, similar to the ones defined in~\cite{GargovPassyTinchev-MODAL-ENVIRONMENT-BOOLEAN-SPECULATIONS-MLIA-1987}, only ``look one step ahead in the structure.'' Such an operator $\paramBox f$ is given by the Boolean function $f$ that ``selects'' a successor world $w'$ of $w$ based on the $R_i$-relationships between $w$ and $w'$ for each accessibility relation $R_i$. As there are only finitely many Boolean functions of a given arity, this only allows to prove our main result for finite partial orders $S$. We also obtain a complete characterization of relative expressiveness and succinctness of modal logics defined in this framework.
  \item ``Several-Step'' operators address worlds that can be reached in arbitrarily many steps. For our result, it suffices to study operators defined by a language $L$ over $\set{1,\dots,n}$: The formula $\paramBox{L}\varphi$, evaluated in a world $w$, requires $\varphi$ to be true in all worlds $w'$ that can be reached from $w$ on a path whose labels form a word in $L$. We show that using \emph{alternation languages} suffices to get arbitrarily complex expressiveness- and succinctness relationships.
 \end{enumerate}

Most of our proofs use formula size games for modal logic as introduced in~\cite{FrenchVanDerHoekIlievKooi-SUCCINCTNESS-MODAL-LOGIC-AI-2013}, based on Adler-Immerman games defined in~\cite{AdlerImmerman-LOWER-BOUND-FORMULA-SIZE-TOCL-2003}. These techniques allow to use games similar to Ehrenfeucht–Fra\"{\i}ss\'e-games to obtain lower bounds on formula size instead of quantifier depth. We adept these games to our generalized settings in the natural way. To the two techniques for establishing lower bounds in Adler-Immerman games mentioned in~\cite{FrenchVanDerHoekIlievKooi-SUCCINCTNESS-MODAL-LOGIC-AI-2013} (namely, \emph{Diverging Pairs} and \emph{Weight Function}), we add a third technique, which is based on a pigeon-hole principle argument. 

The paper is structured as follows: Section~\ref{section:preliminaries} contains the classical definitions of syntax and semantics for modal logics. Section~\ref{section:main results} contains our main results as outlined above. These results are based on a more detailed study of expressiveness and succinctness in the two settings we use, which forms the remainder of the paper: After reviewing formula size games for modal logic introduced in~\cite{FrenchVanDerHoekIlievKooi-SUCCINCTNESS-MODAL-LOGIC-AI-2013} in Section~\ref{section:formula size games}, we present our results on ``One-Step'' and ``Several-Step'' operators in Sections~\ref{proofs:single step operators} and~\ref{proofs:arbitrary step operators}, respectively. We conclude in Section~\ref{sect:conclusion}. All proofs can be found in the appendix.

\section{Preliminaries}\label{section:preliminaries}

We fix an infinite set $V$ of propositional variables. A Kripke model with $n$ modalities is a tuple $M=(W,R_1,\dots,R_n,\Pi)$, where $W$ is a non-empty set of worlds and for each $i\in\set{1,\dots,n}$, $R_i$ is a subset of $W\times W$, and $\Pi\colon P\rightarrow 2^W$ is a propositional assignment. We often simply write $w\in M$ for a world $w\in W$, and $S\subseteq M$ for $S\subseteq W$. A \emph{pointed model} is a pair $(M,w)$ where $M$ is a Kripke model and $w$ is a world of $M$. We also call $w$ the \emph{root} of the pointed model.

The semantics of a modal operator is characterized by the set of worlds that the operator ``adresses'' when evaluated in a pointed model. We formalize this as a \emph{successor selection function}, which is a function $\ssf$ that for each pointed model $(M,w)$ with $n$ modalities returns a set a set $S\subseteq M$. We call $n$ the \emph{arity} of $\ssf$. (Our definition does not rule out mal-formed successor selection functions that do not respect the relational character of modal logic, however all operators we study in this paper are ``well-behaved'' in that sense.)

A successor selection function $\ssf$ naturally defines a modal operator $\paramBox{\ssf}$, by replacing the ``all successors'' of the classical $\Box$-operator with ``all worlds returned by $\ssf$'' (see the formal semantics below). Each set of successor selection functions defines a modal logic as follows (we identify a modal logic with the set of its formulas, as the satisfaction definition will always be standard).

\begin{definition}
 Let $\setSsf$ be a set of successor selection functions. The modal logic $\lmplus{\setSsf}$ is generated by the following grammar:

$$\varphi:=p\ \ \vert \ \ \neg\varphi\ \ \vert \ \ \varphi\vee\varphi \ \ \vert \ \ \paramBox{\ssf}\varphi,$$
where $p\in P$ and $\ssf\in\setSsf$.
\end{definition}

The \emph{size} of a modal formula $\varphi$, denoted $\card\varphi$, is the number of nodes in its tree representation. The semantics definition of $\lmplus{\setSsf}$ is the natural one: 

\begin{definition}
 Let $\varphi$ be an $\lmplus{\setSsf}$-formula, and let $(M,w)$ be a pointed model, where $M=(W,R_1,\dots,R_n,\Pi)$. We define when $\varphi$ is \emph{satisfied} in $w$, written as $M,w\models\varphi$:
 \begin{itemize}
  \item $M,w\models p$ if and only if $w\in\Pi(p)$,
  \item $M,w\models\varphi\vee\psi$ if and only if $M,w\models\varphi$ or $M,w\models\psi$,
  \item $M,w\models\neg\varphi$ if and only if $M,w\not\models\varphi$
  \item $M,w\models\paramBox{\ssf}\varphi$ if $M,w'\models\varphi$ for all $w'\in\ssf(M,w)$.
 \end{itemize}
\end{definition}

For a set $\mathbb M$ of pointed models and a modal formula $\varphi$, we write $\mathbb M\models\varphi$ if $M,w\models\varphi$ for each $(M,w)\in\mathbb M$. 
Formulas $\varphi$ and $\psi$ are \emph{equivalent} if for every pointed model $(M,w)$, we have that $M,w\models\varphi$ if and only if $M,w\models\psi$.

We now define when one modal logic is more expressive or succinct than another. We only state these definitions as far as relevant for this paper, and refer the reader to~\cite{FrenchVanDerHoekIlievKooi-SUCCINCTNESS-MODAL-LOGIC-AI-2013} for an in-depth discussion of these notions.

\begin{definition}
 Let $\setSsf_1$ and $\setSsf_2$ be sets of successor selection functions.
 \begin{itemize}
  \item $\lmplus{\setSsf_2}$ is \emph{at least as expressive} (at least as succinct) as $\lmplus{\setSsf_1}$, written as $\lmplus{\setSsf_1}\leqExpress\lmplus{\setSsf_2}$ ($\lmplus{\setSsf_1}\effRewrite\lmplus{\setSsf_2}$), if for every $\lmplus{\setSsf_1}$-formula $\varphi$, there is an equivalent $\lmplus{\setSsf_2}$-formula $\psi$ (and $\card\psi\leq p(\card\psi)$ for a fixed polynomial $p$). If $\lmplus{\setSsf_1}\leqExpress\lmplus{\setSsf_2}$ and 
  $\lmplus{\setSsf_2}\leqExpress\lmplus{\setSsf_1}$, then $\lmplus{\setSsf_1}$ and $\lmplus{\setSsf_2}$ are \emph{equally expressive}.
  \item If $\lmplus{\setSsf_1}$ and $\lmplus{\setSsf_2}$ are equally expressive, then $\lmplus{\setSsf_1}$ is \emph{exponentially more succinct} than $\lmplus{\setSsf_2}$, if there is a sequence $(\varphi_i)_{i\in\mathbb N}$ of $\lmplus{\setSsf_1}$-formulas such that $\card{\varphi_i}$ is linear in $i$, and there is some $c>1$ such that for each $i$, each $\lmplus{\setSsf_2}$-formula equivalent to $\varphi_i$ has size at least $c^i$.
  \end{itemize}
\end{definition}

Clearly, $\lmplus{\setSsf_1}\not\leqExpress\lmplus{\setSsf_2}$ does not imply that $\lmplus{\setSsf_1}$ is more expressive than $\lmplus{\setSsf_2}$, since $\lmplus{\setSsf_2}\not\leqExpress\lmplus{\setSsf_1}$ can hold simultaneously. If $\lmplus{\setSsf_1}$ is exponentially more succinct than $\lmplus{\setSsf_2}$, then an exponential blow-up in the translation from $\lmplus{\setSsf_1}$ to $\lmplus{\setSsf_2}$ cannot always be avoided, on the other hand, if $\lmplus{\setSsf_1}\effRewrite\lmplus{\setSsf_2}$, then every $\lmplus{\setSsf_1}$-formula can be \emph{succinctly} rewritten into a $\lmplus{\setSsf_2}$-formula. 

One needs to be careful when proving succinctness result via a complexity argument: Unless $\PSPACE=\PTIME$, there is no polynomial-time algorithm converting every closed QBF-formula into a constant formula. However, since each closed QBF-formula is equivalent to either \emph{true} or \emph{false}, the class of closed QBF-formulas is certainly not more succinct than the class of constant formulas. We do not discuss these issues further, since in this paper we will always have that if we compare $\mathcal L_1$ and $\mathcal L_2$ that are equally expressive, then either $\mathcal L_1\effRewrite\mathcal L_2$ and the translation can be computed by a polynomial-time algorithm, or $\mathcal L_1$ is exponentially more succinct than $\mathcal L_2$ in the above, strict sense.

\section{Main Results}\label{section:main results}

We prove that the expressiveness- and succinctness relationships between modal logics can be as complex as any partial order. We show versions of this result in two settings: 
\begin{inparaenum}
 \item For logics defined by successor selection functions $\ssf$ such that whether $w'\in\ssf(M,w)$ only depends on whether $(w,w')\in R_i$ for each accessibility relation $R_i$,
 \item for logics defined by successor selection functions considering paths of arbitrary (finite) length in the model.
\end{inparaenum}

A simple counting argument shows that in the first setting, there is only a finite number of different modal operators, hence for these operators we show that the relationships can be as complex as any finite partial order. In the second setting, we then obtain relationships as complex as any countable partial order.

\subsection{Single Step Operators}\label{sect:single step operators}

In order to prove that the relationships between different modal logics can be arbitrarily complex, we first define a large class of modal logics. All of our logics will be extensions of the classical multi-modal logic $\mathsf{ML}_n$. As a starting point, consider the following modal operators (see also~\cite{VanderhoekIliev-RELATIVE-SUCCINCTNESS-AAMAS-2014-ACM-ENTRY}): For a set $I\subseteq\set{1,\dots,n}$,
\begin{itemize}
 \item $[\forall_I]\varphi$ is true in $w$ if $\varphi$ is true in all $w'$ such that $(w,w')\in R_i$ for some $i\in I$.
 \item $[\cap_I]\varphi$ is true in $w$ if $\varphi$ is true in all $w'$ such that $(w,w')\in R_i$ for all $i\in I$.
\end{itemize}

The first of these operators can be expressed with standard multimodal logic, since $[\forall_I]\varphi$ is equivalent to $\wedge_{i\in I}\Box_i\varphi$. The second one cannot be expressed, since in the standard modal language, we cannot ``address'' a world $w'$ based on whether, for example, $(w,w')\in R_1$ \emph{and} $(w,w')\in R_2$ both hold at the same time. In this section, we consider successor selection functions $\ssf$ that can do exactly this: Whether $w'\in\ssf(M,w)$ depends on the $R_i$-relationships between $w$ and $w_i$ for all $i$ \emph{simultaneously}.

More precisely, we consider $n$-ary successor selection functions $\ssf$ for which the question whether $w'\in\ssf(M,w)$ is described as a Boolean combination of  whether $(w,w')\in R_i$ for each relevant $i$. Such an $\ssf$ is is characterized by a Boolean function $f\colon\set{0,1}^n\rightarrow\set{0,1}$ as follows: For worlds $w,w'$ of a model $M$, we say that $w'$ is an $f$-successor of $w$ if $f(r_1,\dots,r_n)=1$, where $r_i=1$ if $(w,w')\in R_i$, and $r_i=0$ otherwise. Then $f$ yields a successor selection function in the obvious way:

\begin{definition}
 Let $f\colon\set{0,1}^n\rightarrow\set{0,1}$. Then for a pointed model $(M,w)$ with $M=(W,R_1,\dots,R_n,\Pi)$, $\ssf_f(M,w)=\set{w'\ \vert\ w'\mathtext{ is an }f\mathtext{-successor of }w}$.
\end{definition}

We often identify a Booelan function $f$ and the successor selection function $\ssf_f$ defined by $f$. Hence for a set $\mathcal F$ of Boolean functions, we use $\lmF$ to denote the modal logic $\lmplus{\set{\ssf_f\ \vert\ f\in\mathcal F}}$, write $\paramBox{f}$ instead of $\paramBox{\ssf_f}$, etc. The usual multi-model logic with $n$ modalities is obtained as $\mathsf{ML}_n=\lmplus{\set{r_1,\dots,r_n}}$ (we identify a Boolean function with the propositional formula over the variables $\set{r_1,\dots,r_n}$ representing it, hence using the notation above, $\paramBox{r_i}\varphi$ is equivalent to $\Box_i\varphi$). As an example, the above operator $[\forall_I]$ corresponds to the successor selection function $\ssf_{\forall_I}(r_1,\dots,r_n)=\vee_{i\in I}r_i$: It addresses all worlds $w'$ that are an $i$-successor of $w$ for \emph{some} $i\in I$. The operator $[\cap_I]$ similarly corresponds to the successor selection function $\ssf_{\cap_I}(r_1,\dots,r_n)=\wedge_{i\in I}r_i$, as it selects all worlds $w'$ such that $(w,w')\in R_i$ for \emph{all} $i\in I$. 

We now state our main result for modal logics of the form $\lmF$: The expressiveness and succinctness relationships between logics $\lmF$ can be as complex as any finite partial order. 

\begin{restatable}{theorem}{theoremsinglestepmainresult}\label{theorem:single step main result}
 Let $S$ be a finite set, and let $\leq_S$ be a partial order on $S$. Then there exist families of sets of $\lceil\log_2(\card S+1)\rceil$-ary Boolean functions $(\mathcal F_s)_{s\in S}$ and $(\mathcal G_s)_{s\in S}$ such that for each $s,t\in S$, the following holds:
 \begin{enumerate}
  \item $\lmplus{\mathcal F_s}\leqExpress\lmplus{\mathcal F_t}$ if and only if $\lmplus{\mathcal G_s}\effRewrite\lmplus{\mathcal G_t}$ if and only if $s\leq_S t$.
  \item All logics $\lmplus{\mathcal G_s}$ are equally expressive, and if $s\not\leq_S t$, then $\lmplus{\mathcal G_s}$ is exponentially more succinct that $\lmplus{\mathcal G_t}$.
 \end{enumerate}
\end{restatable}

In particular, if $s$ and $t$ are not comparable with respect to $\leq_S$, then $\lmplus{\mathcal F_s}$ is exponentially more succinct than $\lmplus{\mathcal F_t}$ and vice versa, and there are formulas expressible in $\lmplus{\mathcal G_s}$ but not in $\lmplus{\mathcal G_t}$ and vice versa. To prove Theorem~\ref{theorem:single step main result}, we study the expressivity- and succinctness relationship between $\lmF$ and $\lmG$ for different sets $\mathcal F$ and $\mathcal G$ in detail, and obtain a complete characterization that for each $\mathcal F$ and $\mathcal G$ determines the precise relationship between $\lmF$ and $\lmG$ in terms of $\leqExpress$, $\effRewrite$, and exponential succinctness. These results can be found in Section~\ref{proofs:single step operators}.

\subsection{Arbitrary Step Operators}\label{sect:main result:arbitrary step operators}

In this section, we obtain an ``infinite version'' of Theorem~\ref{theorem:single step main result}. As argued above, for a fixed arity $n$, there is only a finite number of modal logics of the form $\lmF$ on Kripke models with $n$ modalities. Hence we consider logics outside of the above framework, i.e., successor selection functions $\ssf$ where whether $w'\in\ssf(M,w)$ does not only depend on whether $(w,w')\in R_i$ for each $i$, but also on longer paths in the model. Natural functions of this form are, e.g., ones returning all worlds reachable on a path of a certain maximal length, or on a path of arbitrary length (which allows to express the transitive closure of the accessibility relations). For our result, it suffices to consider operators of a simple structure, which for a Kripke model with $n$ modalities are given by languages over $\set{1,\dots,n}$. For a word $s=s_1\dots s_l\in\set{1,\dots,n}^*$, we say that a world $w'$ is an $s$-successor of a world $w$ in a model $M$ if there are worlds $w=w_0,w_1,\dots,w_l=w'$ such that for each $i\in\set{1,\dots,l}$, we have that $(w_{i-1},w_i)\in R_{s_i}$. In this case we say that there is an $s$-path from $w$ to $w'$ in $M$, and refer to the $s_i$ as the \emph{labels} of this path. (We omit the model when clear from the context).

A language $L\subseteq\set{1,\dots,n}^*$ defines the successor selection function $\ssf_L(M,w)=\set{w'\in M\ \vert\ w'\mathtext{ is an }s\mathtext{-successor of }w\mathtext{ for some }s\in L}$. Again, we identify a language $L$ and the successor selection function $\ssf_L$, e.g., we write $\paramBox{L}$ instead of $\paramBox{\ssf_L}$, and $\lmplus{\mathcal L}$ for $\lmplus{\set{\ssf_L\ \vert\ L\in\mathcal L}}$, etc. The usual multi-model logic with $n$ modalities is obtained as $\mathsf{ML}_n=\lmplus{\set{\set{1},\dots,\set{n}}}$.

In the sequel, we only consider finite languages. Clearly, for a set $\mathcal L$ of finite languages, every $\lmplus{\mathcal L}$-formula is equivalent to some $\mathsf{ML}_n$-formula, since $\paramBox{L}\varphi$ is equivalent to $\bigwedge_{s=s_1s_2\dots s_k\in L}\Box_{s_1}\Box_{s_2}\dots\Box_{s_k}\varphi$ for a finite language $L$.

Our main result for logics of the form $\lmplus{\mathcal L}$ is an ``infinite version'' of Theorem~\ref{theorem:single step main result}: The succinctness- and expressiveness- relationships between modal logics of the form $\lmplus{\mathcal L}$ can be as complex as any \emph{countable} partial order. For the result, it suffices to consider the bimodal case, i.e., models $(W,R_1,R_2,\Pi)$ with two accessibility relations, and languages over the alphabet $\set{1,2}$. 

\begin{restatable}{theorem}{theoremarbitrarystepmainresult}\label{theorem:arbitrary step main result}
 Let $S$ be a countable set, and let $\leq_S$ be a partial order on $S$. Then there exist families of languages $(\mathcal L_s)_{s\in S}$ and $(\mathcal K_s)_{s\in S}$ over the alphabet $\set{1,2}$ such that for each $s,t\in S$, the following holds:
 \begin{enumerate}
  \item $\lmplus{\mathcal K_s}\leqExpress\lmplus{\mathcal K_t}$ if and only if $\lmplus{\mathcal L_s}\effRewrite\lmplus{\mathcal L_t}$ if and only if $s\leq_S t$.
  \item All logics $\lmplus{\mathcal L_s}$ are equally expressive, and if $s\not\leq_St$, then $\lmplus{\mathcal L_s}$ is exponentially more succinct that $\lmplus{\mathcal L_t}$.
 \end{enumerate}
\end{restatable}

We will give an overview of the proof in Section~\ref{proofs:arbitrary step operators}.

\section{Formula Size Games}\label{section:formula size games}

Our succinctness proofs use modal formula size games introduced in~\cite{FrenchVanDerHoekIlievKooi-SUCCINCTNESS-MODAL-LOGIC-AI-2013} building on Adler-Immerman games~\cite{AdlerImmerman-LOWER-BOUND-FORMULA-SIZE-TOCL-2003}. We review these games in Section~\ref{section:game trees}, and state a variation of their formula-size theorem in Section~\ref{section:fsg theorem}. In Section~\ref{section:pigeonhole}, we introduce the pigeonhole-technique to prove lower bounds on the size of game trees (which then translate to lower bounds on formula size).

\subsection{Game Trees}\label{section:game trees}

The following definition is taken from~\cite{FrenchVanDerHoekIlievKooi-SUCCINCTNESS-MODAL-LOGIC-AI-2013}, except for the straight-forward extension to $\paramBox\ssf$-moves. A game tree represents a formula, where each node $v$ corresponds to a subformula $v_\varphi$ in the natural way. A node $v$ has labels of two kinds: The first label is of the form $\node AB$, where $\mathbb A$ and $\mathbb B$ are classes of pointed models such that $\mathbb A\models v_\varphi$ and $\mathbb B\models\neg v_\varphi$. The second label contains the outmost operator of the formula $v_\varphi$. We simply refer to both labels as ``label,'' it will always be clear whether we refer to the models or the operators. In the following definition, the goal of the single player ``Spoiler'' is to find a formula that is true on all models in $\mathbb A$, and false on all models in $\mathbb B$. Successful plays of Spoiler (called \emph{closed game trees}) directly correspond to such formulas.

\begin{definition}[\cite{FrenchVanDerHoekIlievKooi-SUCCINCTNESS-MODAL-LOGIC-AI-2013}]
 The formula-size game for a set $\setSsf$ of successor selection functions (FSG$(\setSsf)$) on two sets of pointed models $\mathbb A$ and $\mathbb B$ is played as follows: The game begins with a tree containing only the root labelled $\node AB$. In each move of the game, the player (Spoiler) chooses a leaf that is labelled  $\node CD$ for classes $\mathbb C$ and $\mathbb D$ of pointed models and not labelled with a variable, and plays one of the following moves:
 \begin{description}
  \item[atomic move] Spoiler labels the leaf $p$ for a propositional variable $p$ such that $\mathbb C\models p$ and $\mathbb D\models\neg p$.
  \item[not move] Spoiler labels the leaf with $\neg$ and adds a new leaf $\node DC$ as successor.
  \item[or move] Spoiler labels the leaf with $\vee$ and chooses two subsets $\mathbb C_1,\mathbb C_2\subseteq\mathbb C$ with $\mathbb C=\mathbb C_1\cup\mathbb C_2$, then adds successor nodes labelled $\langle \mathbb C_1\circ\mathbb D\rangle$ and $\langle\mathbb C_2\circ\mathbb D\rangle$.
  \item[$\paramBox{\ssf}$-move] Spoiler labels the leaf with $\paramBox{\ssf}$ for some $\ssf\in\setSsf$ and chooses a set $\mathbb D_1$ such that for each $(M,w)\in\mathbb D$, there is some $(M,w')\in\mathbb D_1$ with $w'\in\ssf(M,w)$. A new successor node $\langle\mathbb C_1\circ\mathbb D_1\rangle$ is added to the tree, where $\mathbb C_1=\set{(M,w')\ \vert\ (M,w)\in\mathbb C, w'\in\ssf(M,w)}$.
 \end{description}
 A game tree is \emph{closed} if all of its leafs are labelled with variables.
\end{definition}

By definition, Spoiler cannot play an $\paramBox{\ssf}$-move on a node $\node CD$ if there is some $(M,w)\in\mathbb D$ with $\ssf(M,w)=\emptyset$ (this reflects that $M,w\models\paramBox\ssf\varphi$ for all $\varphi$ in this case). The set $\mathcal T_\setSsf(\node AB)$ contains all closed game trees of FSG$(\setSsf)$ with a root labelled $\node AB$. Spoiler wins the FSG$(\setSsf)$ starting at $\node AB$ in $n$ moves if there is some $T\in\mathcal T_\setSsf(\node AB)$ with exactly $n$ nodes. We usually only write $\mathcal T(\node AB)$ instead of $\mathcal T_\setSsf(\node AB)$ if the set $\setSsf$ is clear from the context. 

\subsection{Formula Size Game Theorem}\label{section:fsg theorem}

The proof of Thoerem~1 from~\cite{FrenchVanDerHoekIlievKooi-SUCCINCTNESS-MODAL-LOGIC-AI-2013} can be generalized in a straight-forward way to give the following result (for completeness, we give the complete proof in Appendix~\ref{sect:proof of prop:fsg theorem}).

\begin{restatable}{theorem}{theoremfsgtheorem}\label{theorem:fsg theorem}{\upshape{\textbf{\cite{FrenchVanDerHoekIlievKooi-SUCCINCTNESS-MODAL-LOGIC-AI-2013}}}}
 Spoiler wins the FSG$(\setSsf)$ starting with $\node AB$ in $k$ moves if and only if there is a formula $\varphi\in\lmplus{\setSsf}$ with $\card\varphi=k$ such that $\mathbb A\models\varphi$ and $\mathbb B\models\neg\varphi$.
\end{restatable}

\subsection{Pigeonhole Principle Technique}\label{section:pigeonhole}

The Formula Size Theorem (Theorem~\ref{theorem:fsg theorem}) allows to prove lower bounds on a $\lmplus{\setSsf}$-formula $\varphi$ by showing a lower bound on the smallest game tree in $\mathcal T_{\setSsf}(\node AB)$, where $\mathbb A\models\varphi$ and $\mathbb\models\neg\varphi$. However, proving a lower bound for game trees is a nontrivial task itself. In~\cite{FrenchVanDerHoekIlievKooi-SUCCINCTNESS-MODAL-LOGIC-AI-2013}, two techniques for proving such a lower bound are mentioned, namely, \emph{Diverging Pairs} and using a \emph{Weight Function}. For our result, we use a  different technique, which is based on a Pigeonhole-like counting argument. The idea is to show that each branch of a formula can only ``cover'' a certain number $c$ of models from $\mathbb A$. From this it then easily follows that the formula must have at least $\frac{\card{\mathbb A}}c$ nodes.

The result uses that formula size games allow the classes of models ``covered'' by each branch of a closed tree (corresponding to a formula) to be simply read off the labels of the leaf of the branch. For a tree $T\in\mathcal T(\node AB)$ and a node $v$ of $T$ labelled $\node CD$, we say that $\mathbb C$ ($\mathbb D$) is the class \emph{corresponding to $\mathbb A$}, if there is an even (odd) number of negations on the path from $T$'s root to $v$, and the class \emph{corresponding to $\mathbb B$} otherwise. 

\begin{restatable}{theorem}{theorempigeonhole}
 \label{theorem:pigeonhole}
 Let $\setSsf$ be a set of successor selection functions. Let $\varphi$ be a formula, let $\mathbb A$ and $\mathbb B$ be sets of pointed models such that $\mathbb A\models\varphi$ and $\mathbb B\models\neg\varphi$. If for every nontrivial leaf $u$ of every closed game tree $T_{\setSsf}\in\mathcal T(\node AB)$, the class of models corresponding to $\mathbb A$ ($\mathbb B$) has size at most $c$, then every $\lmplus{\setSsf}$-formula equivalent to $\varphi$ has size at least $\frac{\card{\mathbb A}}c$ ($\frac{\card{\mathbb B}}c$).
\end{restatable}

\section{Succinctness and Expressiveness for Single-Step Operators}\label{proofs:single step operators}

In this section, we study the expressiveness- and succinctness relationships between logics of the form $\lmF$ for classes $\mathcal F$ of Boolean functions. In particular, these results allow us to prove the above Theorem~\ref{theorem:single step main result}. We first consider expressiveness. The following result completely answers the question in which case $\lmG\leqExpress\lmF$ holds: 

\begin{restatable}{theorem}{theoremonestepexpressivenessorequivalence}\label{theorem:1 step expressiveness -- or equivalence}
 Let $\mathcal F$ and $\mathcal G$ be sets of $n$-ary Boolean functions. Then the following are equivalent:
 \begin{enumerate}
  \item\label{theorem:1 step expressiveness -- or equivalence:expressiveness} $\lmG\leqExpress\lmF$,
  \item\label{theorem:1 step expressiveness -- or equivalence:or} for each $g\in\mathcal G$, there is a set $S\subseteq\mathcal F$ such that $g\equiv\bigvee_{f\in S}f$.
 \end{enumerate}
\end{restatable}

For example, the theorem implies the result mentioned in Section~\ref{sect:single step operators} that $[\forall_I]$ can be expressed with the standard operators $\Box_1$ and $\Box_2$, but $[\cap_I]$ cannot (recall that $[\forall_I]$ corresponds to $\vee_{i\in I}r_i$, and $[\cap_I]$ to $\wedge_{i\in I}r_i$).

Theorem~\ref{theorem:1 step expressiveness -- or equivalence} is proved using standard bisimulation techniques (see Appendix~\ref{sect:proof theorem:1 step expressiveness -- or equivalence}), which show that a specific formula cannot be expressed in a logic $\lmF$. We now consider succinctness. The following theorem says that, given sets $\mathcal F$ and $\mathcal G$ of Boolean functions such that $\lmF$ and $\lmG$ are equally expressive, $\lmG$ is \emph{always} exponentially more succinct than $\lmF$, except for the trivial case when $\mathcal G\subseteq\mathcal F$.

The proof of the theorem indeed shows the slightly stronger result that even if $\lmF$ and $\lmG$ are not equally expressive, but $\mathcal G$ contains a function that is a disjunction of functions in $\mathcal F$ but is not an element of $\mathcal F$ itself (and hence, due to Theorem~\ref{theorem:1 step expressiveness -- or equivalence}, $\paramBox{g}$ is not expressible in $\lmF$), then $\lmG$ is exponentially more succinct than $\lmF$ (with a slightly more general definition of this notion that also covers modal languages with different expressive power). This implies that the relation $\effRewrite$ restricted to logics of the form $\lmF$ is antisymmetric, and hence a partial order.

A special case of our result was shown in~\cite{FrenchVanDerHoekIlievKooi-SUCCINCTNESS-MODAL-LOGIC-AI-2013}, where the authors prove that the logic $[\forall_{1,2}]\mathsf{ML}$ is exponentially more succinct than $\mathsf{ML}_2$. Using our notation, these logics are $\lmG$ and $\lmF$ with $\mathcal F=\set{r_1,r_2}$ and $\mathcal G=\set{r_1,r_2,r_1\vee r_2}$. 

\begin{restatable}{theorem}{theoremfinitestepsuccinctness}\label{theorem:finite step succinctness}
 Let $\mathcal F$ and $\mathcal G$ be sets of Boolean functions such that $\lmF$ and $\lmG$ are equally expressive and $\mathcal G\nsubseteq\mathcal F$. Then $\lmG$ is exponentially more succinct than $\mathcal F$.
\end{restatable}

The proof of Theorem~\ref{theorem:finite step succinctness} uses an extension of the technique used to prove the above-mentioned result in~\cite{FrenchVanDerHoekIlievKooi-SUCCINCTNESS-MODAL-LOGIC-AI-2013}. The main additions we make to their construction are ``false paths'' in the models that stop Spoiler from using operators of the form $\paramBox{f}$ where $f$ is not one of the functions from $\mathcal F$ appearing in the disjunctive definition of $g$, and a generalization of edges labelled with $1$ and $2$ to edges labelled with appropriate Boolean combinations of the involved modalities. Finally, instead of the diverging pairs technique, we use the pigeonhole technique to prove the lower bound on the game tree size. The proof can be found in Appendix~\ref{appendix:proof of theorem:finite step succinctness} (In Section~\ref{proofs:arbitrary step operators}, we give a more detailed presentation of an application of the pigeonhole technique.)

\section{Succinctness and Expressiveness for Arbitrary-Step Operators}\label{proofs:arbitrary step operators}

In this section, we give an overview of the proof of Theorem~\ref{theorem:arbitrary step main result}. In particular, we define the sets of languages $\mathcal L_s$ mentioned in the statement of the theorem as sets of \emph{alternation languages} (Section~\ref{sect:alternation language introduction}). We then study the relationships between $\lmplus{\mathcal L_1}$ and $\lmplus{\mathcal L_2}$ for sets $\mathcal L_1$ and $\mathcal L_2$ of alternation languages in detail. Due to the page limit, we only give the construction (Section~\ref{sect:model construction}) and state its main technical properties (Section~\ref{sect:formula size games alternation}), and the consequences for expressiveness (Section~\ref{section:alternation expressiveness}) and succinctness (Section~\ref{section:alternation succinctness}). The technical proofs are deferred to Appendix~\ref{appendix:arbitrary step proofs}.

To prove Theorem~\ref{theorem:arbitrary step main result}, it is enough to consider bimodal logics, i.e., models with two accessibility relations and thus languages over the alphabet $\set{1,2}$. We therefore only consider this case in the remainder of this section.

\subsection{Alternation Languages}\label{sect:alternation language introduction}

Let $\ell\ge1$ be a natural number. A word $s=s_1\dots s_\ell\in\set{1,2}^\ell$ is \emph{alternating} if for each $i\in\set{1,\dots,\ell-1}$, $s_i\neq s_{i+1}$. There are exactly two alternating words of length $\ell$, namely $\altword 1$, starting with $1$, and $\altword 2$, starting with $2$. The \emph{alternation language of length $\ell$}, denoted with $A_\ell$, is the set $\set{\altword 1,\altword 2}$. 

Following the definitions in Section~\ref{sect:main result:arbitrary step operators}, the language $A_\ell$ defines the modal operator $\paramBox{A_\ell}$, where $\paramBox{A_\ell}\varphi$ requires $\varphi$ to be true in all worlds reachable on a path whose labels form an alternating word of length $\ell$. This operator is natural in an epistemic setting, where it can be read as ``$A$ knows that $B$ knows that $A$ knows that $B$ knows \dots'' and vice versa, to the $\ell$-th degree.

The iterated application of the operator $\paramBox{A_\ell}$, denoted as usual with $\paramBox{A_\ell}^i$, addresses all worlds accessible on a path whose labels form a sequence of $i$ words from $A_\ell$. To be able to address the specific alternating words in this sequence, we extend the notation $\altword 1$ and $\altword 2$ above:  For a word $s=s_1\dots s_i\in\set{1,2}^*$, with $\altword s$ we denote the word $\altword{s_1}\altword{s_2}\dots\altword{s_i}$, i.e., the word consisting of $i$ alternating words of length $\ell$, where the $j$-th of these words starts with $s_j$. 

For a set $I\subseteq\mathbb N$, let $\lmA I$ denote the logic $\lmplus{\set{A_\ell\ \vert\ \ell\in I}}$, and let $\lmAp I$ denote the logic $\lmplus{ \set{A_\ell\ \vert\ \ell\in I}\cup\set{\set1,\set2}}$. Hence in the logic $\lmA I$, all operators $\paramBox{A_\ell}$ with $\ell\in I$ are allowed, the logic $\lmAp I$ additionally allows the classical operators $\Box_1$ and $\Box_2$. Since all logics $\lmAp I$ contain $\Box_1$ and $\Box_2$, and all involved languages are finite, all $\lmAp I$ have the same expressive power, namely that of classical bimodal logic. This is not true for the logics $\lmA I$, as we will see below.

Our main result on alternation languages states that if $I_1$ is not a subset of $I_2$, then $\lmAp{I_1}$ is exponentially more succinct than $\lmAp{I_2}$, and $\lmA{I_1}$ contains formulas that are not expressible in $\lmA{I_2}$. This result is the key step to proving Theorem~\ref{theorem:arbitrary step main result}.

\begin{restatable}{theorem}{theoremmainalternationresultsubsetequivalence}\label{theorem:main alternation result: subset equivalence}
 Let $I_1,I_2\subseteq\mathbb N$ with $I_1\nsubseteq I_2$. Then 
 \begin{enumerate}
  \item\label{enum:succinctness} $\lmAp{I_1}$ is exponentially more succinct than $\lmAp{I_2}$, and
  \item\label{enum:expressiveness} there is an $\lmA{I_1}$-formula for which there is no equivalent $\lmA{I_2}$-formula.
 \end{enumerate}
\end{restatable}

Note that, in contrast to the situation for logics of the form $\lmF$ for a set of Boolean functions $\mathcal F$ (see Section~\ref{proofs:single step operators}), we do not get the corresponding result that the relation $\effRewrite$, restricted to logics of the form $\lmplus{\mathcal L}$ is antisymmetric. The reason for this is that by adding an operator $\paramBox{L}$ for a \emph{singleton} language to a logic containing both classical operators $\Box_1$ and $\Box_2$ changes neither expressiveness nor succinctness of the logic; hence an arbitrary number of logics equivalent to $\mathsf{ML}_n$ in expressiveness and succinctness can be defined in this way.

\subsection{Model Construction}\label{sect:model construction}

Our main result about alternation languages, and the main ingredient to the proof of Theorem~\ref{theorem:main alternation result: subset equivalence}, is that if $\ell\notin I$, then every $\lmAp I$-formula $\psi$ equivalent to $\paramBox{A_\ell}^ip$ is exponentially large (in $i$), and there is no $\lmA I$-formula equivalent to $\paramBox{A_\ell}p$. In the following discussion, we focus on the more involved succinctness result.

We start by defining the models on which we play the formula size game: For each $i$ and $\ell\ge1$, we define classes of pointed models $\modelClassA$ and $\modelClassB$ such that $\modelClassA\models\paramBox{A_\ell}^ip$, and $\modelClassB\models\neg\paramBox{A_\ell}^ip$.

These classes are defined in three steps:
\begin{enumerate}
 \item For each $i\in\mathbb N$, we define a \emph{base model} $\baseA$, and for each $s\in\set{1,2}^i$, a \emph{base model} $\baseB$.
 \item We then extend the models $\baseA$ and $\baseB$ to $\starA$ and $\starB$. The extension consists of adding a ``trap'' to the models which the Spoiler candidate may never choose as successors in the $\Box_\ssf$-step of the formula size games. This 
 \item for each $i\in\mathbb N$, we define the above-mentioned classes. Since for each such $i$, there is only a single model $\starA$, we wimply identify the model $\starA$ with the class $\set{\starA}$ and write only $\modelClassA$ for this class. On the other hand, the class $\modelClassB$ contains all models $\starB$ with $s\in\set{1,2}^i$.
\end{enumerate}

The main idea of the definition of our models is to ensure that each branch of each tree $T\in\mathcal T_\setSsf(\angNode{\modelClassA}{\modelClassB})$ corresponding to a smallest formula $\psi$ as above can only cover a restricted number of models (namely, at most $2^{\frac i2}$ models). This allows us to use the pigeonhole technique (Theorem~\ref{theorem:pigeonhole}) to prove our succinctness result.

For $i,\ell\in\mathbb N$, let $s=s_1\dots s_i\in\set{1,2}^i$. The main idea of the following models is that in order to ``cover'' all models $\starB$, Spoiler needs to exhibit, for each of them, an $A_\ell^i$-successor of the root in which the variable $p$ is false. The ``traps'' in the models $\starB$ ensure that the only path for which this is true is the word $\altword s$. Therefore, Spoiler needs to play a strategy that covers at least all strings of the form $\altword s$. The model $\starA$ forces Spoiler to only cover strings of this form. Hence Spoiler must cover \emph{exactly} all strings $\altword s$ where $s\in\set{1,2}^i$, which cannot be done succinctly without the operator $\paramBox{A_\ell}$. The ``base models'' models $\baseA$ and $\baseB$ are defined as in the following picture:  
 
                 
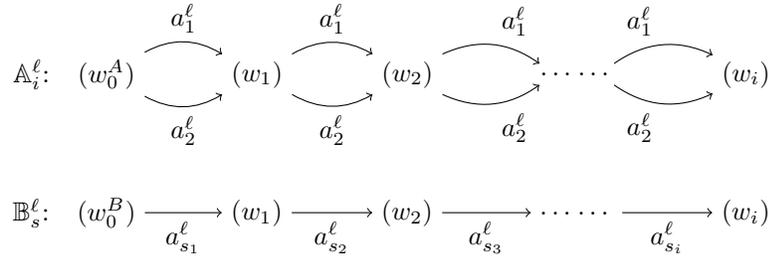
\begin{tikzpicture}
                    \node at (-1,0)         {$\baseA$:};
		    \node at (0,0)  (w0)    {$(w^A_0)$};
		    \node at (2,0)  (w1)    {$(w_1)$};
		    \node at (4,0)  (w2)    {$(w_2)$};
		    \node at (6,0)  (dots1) {$\dots$};
		    \node at (6.5,0)  (dots2) {$\dots$};
		    \node at (8.5,0) (wi)    {$(w_i)$};
  
		    \draw (w0)    edge[bend right,->] node [auto=right] {$\altword 2$} (w1);
		    \draw (w0)    edge[bend left,->]  node [auto=left]  {$\altword 1$} (w1);
		    \draw (w1)    edge[bend right,->] node [auto=right] {$\altword 2$} (w2);
		    \draw (w1)    edge[bend left,->]  node [auto=left]  {$\altword 1$} (w2);
		    \draw (w2)    edge[bend right,->] node [auto=right] {$\altword 2$} (dots1);
		    \draw (w2)    edge[bend left,->]  node [auto=left]  {$\altword 1$} (dots1);
		    \draw (dots2) edge[bend right,->] node [auto=right] {$\altword 2$} (wi);
		    \draw (dots2) edge[bend left,->]  node [auto=left]  {$\altword 1$} (wi);
                 \end{tikzpicture} \\
                 
                 \begin{tikzpicture}
		    \node at (-1,0)         {$\baseB$:};
		    \node at (0,0)  (w0)    {$(w^B_0)$};
		    \node at (2,0)  (w1)    {$(w_1)$};
		    \node at (4,0)  (w2)    {$(w_2)$};
		    \node at (6,0)  (dots1) {$\dots$};
		    \node at (6.5,0)  (dots2) {$\dots$};
		    \node at (8.5,0) (wi)    {$(w_i)$};
  
		    \draw (w0)    edge[->] node [auto=right] {$\altword{s_1}$} (w1);
		    \draw (w1)    edge[->] node [auto=right] {$\altword{s_2}$} (w2);
		    \draw (w2)    edge[->] node [auto=right] {$\altword{s_3}$} (dots1);
		    \draw (dots2) edge[->] node [auto=right] {$\altword{s_i}$} (wi);
	         \end{tikzpicture}

 An edge labelled $\altword j$ for $j\in\set{1,2}$ between $w_m$ and $w_{m+1}$ indicates that $w_{m+1}$ is an $\altword j$-successor of $w_m$. This is achieved by intermediate worlds (not shown in the picture) $w_m=u_0,\dots,u_\ell=w_{m+1}$ such that $(u_{p-1},u_p)\in R_{\altword j[p]}$ for each relevant $p$. The propositional variable $p$ is false in all worlds, except for the world $w_i$ of $\baseA$. The root of $\baseA$ ($\baseB$) is $w_0^A$ ($w_0^B$). By construction, each world $u$ in $\baseA$ and $\baseB$ has a unique distance from the model's root. We denote this distance with $\depth u$.
 
 We obtain the ``extended'' models $\starA$ and $\starB$ from $\baseA$ and $\baseB$ as follows: To both models, we add a new node $\wtrap$, which is a reflexive singleton (i.e., a world with an $1$- and a $2$-edge to itself). In $\starA$, the variable $p$ is false in $\wtrap$, in $\starB$, the variable is true in $\wtrap$. For each world $w$ of $\baseA$ ($\baseB$) that does not have a $j$-successor for some $j\in\set{1,2}$, we add a $j$-edge leading to the world $\wtrap$ of the respective model. These edges are ``false paths,'' since the reflexive singleton of $\starA$ does not allow Spoiler to prove that all relevant paths end in a world satisfying $p$, and the singleton in $\starB$ does not allow Spoiler to find a path to a world where $p$ is false. Hence these ``false paths'' are never taken in a closed game tree that corresponds to a minimal formula.
  
 Our classes of models now contain all models constructed in the above way: For $i\ge1$, we identify $\modelClassA$ with the singleton $\set{\starA}$, and define  $\modelClassB=\set{\starB\ \vert\ s\in\set{1,2}^i}$. Then $\modelClassA\models\paramBox{A_\ell}^ip$, and $\modelClassB\models\neg\paramBox{A_\ell}^ip$:

\begin{itemize}
 \item Every path made up of $i$ alternating words of length $\ell$ starting at $w^A_0$ in $\starA$ leads to the world $w_i$ of $\baseA$, where $p$ is true (no such path ends in the reflexive singleton).
 \item For $s=s_1\dots s_i\in\set{1,2}^i$, the world $w_i$, in $\starB$, is an $\altword s$-successor of $w^B_0$ and does not satisfy $p$.
\end{itemize}

\subsection{Formula Size Games on our Models}\label{sect:formula size games alternation}

We now state a few technical results on formula size games on closed game trees in $\mathcal T(\angNode\modelClassA\modelClassB)$. Essentially, these results say that Spoiler indeed needs to play a strategy as intended by the definition of our models. Recall that our goal is to show that if $\ell\notin I$, then every $\lmAp{I}$-formula $\psi$ equivalent to $\paramBox{A_\ell}^ip$ must be of exponential size. In the following, we fix a smallest such formula $\psi$, and consider the game tree that corresponds to the evaluation of $\psi$ on the classes of models $\modelClassA$ and $\modelClassB$ in the following way:
For a formula $\psi$ and classes $\mathbb A$ and $\mathbb B$ of pointed models with $\mathbb A\models\psi$ and $\mathbb B\models\neg\psi$, let $\paramParseTree{\psi}{\mathbb A}{\mathbb B}$ be the closed game tree obtained from following the strategy corresponding to $\psi$ on the starting node $\node AB$. Clearly, $\paramParseTree{\psi}{\mathbb A}{\mathbb B}\in\mathcal T_\setSsf(\node AB)$ if $\setSsf$ contains at least all succesor selection functions appearing in $\psi$.

We first show that the formula $\psi$ indeed must indeed avoid the ``traps'' added to the models, as intended:

\begin{restatable}{lemma}{lemmaminimalformulagivestrapavoidingtree}\label{lemma:minimal formula gives trap avoiding tree}
 Let $I\subseteq\mathbb N$, let $\psi$ be a minimal $\lmAp{I}$-formula equivalent to $\paramBox{A_\ell}^ip$, let $v$ be a node of 
 $\parseTree$ labelled $\node AB$, and let $(M,w)\in\mathbb A\cup\mathbb B$. Then $w\neq\wtrap$.
\end{restatable}

Our next result is that in the formula $\psi_i$, operators $\paramBox{A_{\ell'}}$ can only appear in depths that are multiples of $\ell$. The proof uses that each path that is not a prefix of a word in $(A_\ell)^i$ leads to $\wtrap$ in the models from $\modelClassA$ and $\modelClassB$.
For a node $v$ of a tree $T$, with $\boxlabels v$, we denote the sequence of successor selection functions appearing in $\paramBox{.}$ operators on the path from $T$'s root to $v$, excluding the label of $v$ itself. (We do not make $T$ explicit in the notation, this will always be clear from the context, and again identify $L$ and $\ssf_L$ for a language $L$).

We say that a language $\emptyset\neq L\subseteq\set{1,2}^*$ is \emph{length}-uniform if there is some $i$ such that $L\subseteq\set{1,2}^i$, i.e., all words in $L$ have the same length. We denote this length $i$ with $\dcard L$. Clearly, the class of length-uniform languages is closed under concatenation, and all languages $L$ we consider in this section (the alternating languages $A_\ell$ and the languages $\set1$ and $\set2$) are length-uniform. For a node $v\in\parseTree$, and a string $s\in\set{1,2}^i$, we say that $v$ \emph{covers} $s$, if one of the classes of models with which $v$ is labelled contains a model $(\starB,w)$ for some $w\in\starB$. As discussed before, we will show that each leaf $v\in\parseTree$ can only cover a restricted number of strings $s$.

\begin{restatable}{lemma}{lemmaalternationsplittingonlyatmultiplesofl}\label{lemma:alternation:splitting only at multilpes of l}
 Let $\psi$ be a minimal $\lmAp{I}$-formula equivalent to $\paramBox{A_\ell}^ip$. Let $v\in\parseTree$. Let $\boxlabels{v}=L_1\dots L_{m-1}A_{\ell'}$, where each $L_i$ is length-uniform and $\ell'\in\mathbb N$. Then $\dcard{L_1\circ\dots\circ L_{m-1}}$ is a multiple of $\ell$.
\end{restatable}

The next result, again following from the ``false paths'' in the construction, is that in $\mathcal T(\angNode{\starA}{\starB})$, Spoiler indeed needs to play the intended strategy, namely, following exactly the path $\altword{s}$ in the model $\starB$. This follows from the above, since there is only one path avoiding $\wtrap$ in model $\starB$, namely the path $\altword s$. Hence Spoiler has to play a sequence of languages covering $\altword s$ for each $s\in\set{1,2}^i$.

\begin{restatable}{lemma}{lemmaalternationalsalwaysinlanguage}\label{lemma:alternation: a-l-s always in language}
 Let $\psi$ be a minimal $\lmAp{I}$-formula equivalent to $\paramBox{A_\ell}^ip$. Let $\boxlabels{v}=L_1\dots L_m$, where each $L_i$ is length-uniform. Let $L:=L_1\circ\dots\circ L_m$, let $d:=\dcard L$. Then $\altword s[1\dots d]\in L$ for each $s\in\set{1,2}^i$ such that $v$ covers $s$.
\end{restatable}

The above two results can now be used to prove that we can indeed apply the pigeonhole technique (Theorem~\ref{theorem:pigeonhole}). For this we show, in the final two results of this section, that each branch of $\parseTree$ only covers a restricted number of strings $s\in\set{1,2}^i$. The reason for this is that the application of each available operator $\Box_1$, $\Box_2$ and $\paramBox{A_{\ell'}}$ for $\ell'\neq\ell$ comes with the ``cost'' of excluding a significant set of values $s$ that the corresponding branch of the formula covers.

The first of these two results addresses the case where a branch a the formula equivalent to $\paramBox{A_\ell}^i$ uses an operator $A_{\ell'}$, where $\ell'$ is not a multiple of $\ell$. Due to the above Lemma~\ref{lemma:alternation:splitting only at multilpes of l}, such occurrances are restricted to modal depths which themselves are a multiple of $\ell$. Therefore, \emph{immediately after} such an application, no operator $\paramBox{A_{\ell''}}$ can appear, and a classical operator $\Box_j$ for $j\in\set{1,2}$ must be used. Hence such a branch can only ``cover'' paths in the model that have the symbol $j$ at the next position, which is only true for half of the words in $(A_\ell)^i$. Hence each such application of a modal operator in a branch halves the number of strings $s$ covered with this branch. One can also derive the expressiveness part of Theorem~\ref{theorem:main alternation result: subset equivalence} from this lemma, since in the logic $\lmA I$, the required classical operators simply are not available.

\begin{restatable}{lemma}{lemmalprimenotmultipleoflforcesconstant}\label{lemma:l' not multiple of l forces constant}
 Let $\psi$ be a minimal $\lmAp{I}$-formula equivalent to $\paramBox{A_\ell}^ip$. Let $v\in\parseTree$ such that $v$ covers $s_1$ and $s_2$. Let $\boxlabels{v}=L_1\dots L_m$, and let $L_i=A_{f\cdot\ell+q}$ for $f\ge 0$ and $1\leq q<\ell$. Then $s_1[u]=s_2[u]$, where $u=\frac1\ell\cdot\dcard{L_1\circ\dots\circ L_{i-1}}+f+1$.
\end{restatable}

Our last result in this section addresses the case that an operator $A_{\ell'}$ with some ``large'' $\ell'$ (i.e., larger than $\ell$) appears. In this case, the operator $A_{\ell'}$ only addresses worlds that are reachable on a path $\altword s$ that has a sequence of $\ell'$ consecutive alternations. This directly implies restrictions on the string $s$ as follows:

\begin{restatable}{lemma}{lemmalprimegreaterthanlforcesidentities}\label{lemma:l' > l forces identities}
 Let $\psi$ be a minimal $\lmAp{I}$-formula equivalent to $\paramBox{A_\ell}^ip$. Let $v\in\parseTree$ such that $v$ covers $s\in\set{1,2}^i$, let $\boxlabels{v}=L_1\dots L_m$, and let $L_i=A_{\ell'}$. Then, for $u=\frac1\ell\dcard{L_1\dots L_{i-1}}$ and all $j$ with $1\leq j<\frac{\ell'}\ell$, we have:
 
 \begin{itemize}
  \item If $\ell$ is even, then $s[u+j]=s[u+j+1]$.
  \item If $\ell$ is odd, then $s[u+j]=3-s[u+j+1]$.
 \end{itemize}
\end{restatable}

\subsection{Applications for Expressiveness}\label{section:alternation expressiveness}

We now obtain the expressiveness part of Theorem~\ref{theorem:arbitrary step main result}: If $\ell\notin I$, then the logic $\lmA I$ cannot express the formula $\paramBox{A_\ell}p$. (Recall that, unlike $\lmAp I$, the logic $\lmA I$ does not contain the standard modal operators $\Box_1$ and $\Box_2$.) As discussed earlier, the result  follows from Lemma~\ref{lemma:alternation:splitting only at multilpes of l} with a syntactic argument.

\begin{restatable}{theorem}{theoremalternationexpressiveness}\label{theorem:alternation expressiveness}
 Lei $I\subseteq\mathbb N$, let $\ell\in\mathbb N$ with $\ell\notin I$. Then there is no formula $\varphi\in\lmA I$ that is equivalent to $\paramBox{A_\ell}p$.
\end{restatable}

\subsection{Applications for Succinctness}\label{section:alternation succinctness}

We now show that if $\ell\notin I$, then in every mimimal $\lmAp I$-formula $\psi$ equivalent to $\paramBox{A_\ell}^ip$ for an even number $i$, each branch can cover at most $2^{\frac i2}$ models. This upper bound is tight, since the formula $\paramBox{A_{2\ell}}^{\frac i2}p$ covers $2^{\frac i2}$ strings, namely each  $s\in\set{11,22}^{\frac i2}$.

\begin{restatable}{lemma}{lemmarestrictiononsetsbyonebranch}
\label{lemma:restriction on sets by one branch}
 Let $\psi$ be a minimal $\lmAp{I}$-formula equivalent to $\paramBox{A_\ell}^ip$, where $i$ is even and $\ell\notin I$. Let $v$ be a leaf of $\parseTree$, and let $S\subseteq\set{1,2}^i$ be the set of strings $s$ such that $v$ covers $s$. Then $\card S\leq 2^{\frac i2}$. 
\end{restatable}

From Lemma~\ref{lemma:restriction on sets by one branch} and the pigeonhole technique, we directly obtain obtain the following result.

\begin{restatable}{theorem}{theoremalternationsuccinctness}\label{theorem:alternation succinctness}
Let $I\subseteq\mathbb N$ with $\ell\notin I$, and let $\psi\in\lmA I$ be equivalent to $\paramBox{A_\ell}^ip$, where $i$ is an even number. Then $\card\psi\ge2^{\frac i2}$.
\end{restatable} 

\section{Conclusion}\label{sect:conclusion}

We proved that the expressiveness- and succinctness relationships between modal logics can be as complex as any finite or countable partial order. In the first setting we studied logics of the form $\lmF$ for a set $\mathcal F$ of Boolean functions. Here we obtained a complete characterization of the relative expressiveness and succinctness of logics $\lmF$ and $\lmG$. It is an interesting open question to obtain a similar complete characterization for the second setting, i.e., to answer completely the question for which sets of languages $\mathcal L$ and $\mathcal K$ we have that $\lmplus{\mathcal L}$ is more succinct or more expressive than $\lmplus{\mathcal K}$.

\begin{appendix}
 
\section{Omitted Results and Definition}

\subsection{Facts about Formula Size Games}

In~\cite{FrenchVanDerHoekIlievKooi-SUCCINCTNESS-MODAL-LOGIC-AI-2013}, it was shown that if $v$ is a node of a closed game tree labelled $\node CD$, then $\mathbb C$ and $\mathbb D$ do not contain bisimilar pointed models. In particular, this implies the following:

\begin{proposition}\label{prop:closed game tree disjoint sets}
 Let $v$ be a node in a closed game tree $T$ labelled $\node CD$. Then $\mathbb C\cap\mathbb D=\emptyset$.
\end{proposition}

The definition of formula-size games immediately leads to the following easy property:

\begin{restatable}{lemma}{lemmafsgpathproperty}\label{lemma:fsg path property}
 Let $T\in\mathcal T{\node AB}$, and let $v\in T$ be a node labelled with $\node CD$, and let $\boxlabels v=\ssf_1\dots\ssf_m$. Let $(M,w)\in\mathbb C$, and let $(M',w')\in\mathbb D$. Then there exist worlds $w_0,w_1,\dots,w_m=w\in M$ and $w_0',w_1',\dots,w_m'=w'\in M'$ such that for each $j\in\set{1,\dots,m}$, we have that $w_{j}\in\ssf_j(M,w_{j-1})$ and $w_{j}'\in\ssf_j(M',w_{j-1}')$ and
 \begin{enumerate}
  \item $(M,w_0)$ is an element of the class corresponding to $\mathbb A$,
  \item $(M',w_0')$ is an element of the class corresponding to $\mathbb B$.
 \end{enumerate}
\end{restatable}

\begin{proof}
 Let $r$ be the root of $T$. We show the claim by induction on the length of the path from $r$ to $v$. In the base case, we have that $r=v$ and therefore $\node AB=\node CD$, and the path contains an even number of negations. Hence the claim follows trivially.
 
 Now assume that the claim is true for the unique predecessor nove $v_0$ of $v$ in $T$, where $v_0$ is labelled with $\node{C_0}{D_0}$.
 
 We make a case distinction depending on the label of $v_0$. Note that since $v_0$ is not a leaf, $v_0$ cannot be labelled with an atomic proposition.
 
 \begin{itemize}
  \item If $v_0$ is labelled with $\neg$, then in particular, $\boxlabels{v_0}=\boxlabels{v}$, $\mathbb C_0=\mathbb D$, and $\mathbb D_0=\mathbb C$. The claim follows trivially by induction.
  \item If $v_0$ is labelled with $\vee$, then $\mathbb C\subseteq\mathbb C_0$ and $\mathbb D=\mathbb D_0$. 
  \item Finally, let $v_0$ be labelled with $\paramBox{\ssf_m}$. It then follows that $\boxlabels{v_0}=\ssf_1\dots\ssf_{m-1}$. 
  
   By definition of the game, we know that for each $(M,w)\in\mathbb C$, there is a world $w_{m-1}\in M$ with $(M,w_{m-1})\in\mathbb C_0$ and $w\in\ssf_m(M,w_{m-1})$, and for each $(M',w')\in\mathbb D$, there is a world $w'_{m-1}\in M'$ with $(M',w'_{m-1})\in\mathbb D_0$ and $w'\in\ssf_m(M',w'_{m-1})$. Due to induction, worlds $w_0,w_1,\dots,w_{m-2}\in M$ and $w'_0,w'_1,\dots,w'_{m-2}\in M'$ can be chosen with the required properties. This concludes the proof.
 \end{itemize}
\end{proof}

\subsection{Bisimulations}\label{sect:appendix:bisimulations}

\begin{definition}
 Let $M_1=(W^1,R^1_1,\dots,R^1_n,\Pi^1)$ and $M_2=(W^2,R^2_1,\dots,R^2_n,\Pi^2)$ be Kripke models, let $\setSsf$ be a set of successor selection functions. A relation $Z\subseteq W^1\times W^2$ is an \emph{\setSsf-bisimulation between $M_1$ and $M_2$} if for all $(w_1,w_2)\in Z$, the following holds:
 \begin{itemize}
  \item for all $p\in P$, we have that $w_1\in\Pi^1(p)$ if and only if $w_2\in\Pi^2(p)$,
  \item (forward condition) for all $\ssf\in\setSsf$ and all $w_1'\in\ssf(M_1,w_1)$, there is some $w_2'\in\ssf(M_2,w_2)$ such that $(w_1',w_2')\in Z$,
  \item (back condition) for all $\ssf\in\setSsf$ and all $w_2'\in\ssf(M_2,w_2)$, there is some $w_1'\in\ssf(M_1,w_1)$ such that $(w_1',w_2')\in Z$.
 \end{itemize}
\end{definition}

The following is easy to see:

\begin{proposition}\label{prop:invariance under f1 ... fk bisimulation}
 Let $M_1$ and $M_2$ be Kripke models, let $\setSsf$ be a set of successor selection functions, let $Z$ be a $\setSsf$-bisimulation between $M_1$ and $M_2$, and let $(w^1,w^2)\in Z$. Then for each formula $\varphi$ of $\lmplus{\setSsf}$, we have that $M_1,w^1\models\varphi$ if and only if $M_2,w^2\models\varphi$.
\end{proposition}

\begin{proof}
 As usual by induction on the formula. The base case where $\varphi$ is a propositional variable is trivial, the cases where $\varphi$ is a disjunction or a negation follow by induction. Hence let $\varphi=\paramBox{\ssf}\psi$, and let $M_1,w_1\models\varphi$. To show that $M_2,w_2\models\varphi$, let $w_2'\in\ssf(M_2,w_2)$. Since $(w_1,w_2)\in Z$ and $Z$ is a $\setSsf$-bisimulation, there is some world $w_1'\in\ssf(M_1,w_1)$ with $(w_1',w_2')\in Z$. Since $M_1,w^1\models\varphi$, it follows that $M_1,w_1'\models\psi$, and hence due to induction we have that $M_2,w_2'\models\psi$. Therefore, it follows that $M_2,w_2'\models\varphi$. The converse is symmetric.
\end{proof}
 
\section{Proofs of Results in Main Paper}

\subsection{Extensions of Formula Size Games}

\subsubsection{Proof of Theorem~\ref{theorem:fsg theorem}}\label{sect:proof of prop:fsg theorem}

 The proof of Theorem~\ref{theorem:fsg theorem} is an adaptation of the corresponding result in~\cite{FrenchVanDerHoekIlievKooi-SUCCINCTNESS-MODAL-LOGIC-AI-2013}, the extension to arbitrary modal operators is  straight-forward.
 
 \theoremfsgtheorem*
 
\begin{proof}
 First assume that there is a formula $\varphi$ of size $k$ such that $\mathbb A\models\varphi$ and $\mathbb B\models\neg\varphi$. We prove by induction on the construction of $\varphi$ that Spoiler can win the FSG starting with $\langle\mathbb A\circ\mathbb B\rangle$ in $k$ moves by using the strategy encoded in the formula $\varphi$.
 
 If $\varphi$ is a propositional variable $p$, then clearly Spoiler can win by playing the move $p$. 
 
 If $\varphi=\neg\psi$ for a modal formula $\psi$, then Spoiler plays the not-move, which results in a node labelled $\node BA$. Since $\mathbb A\models\varphi$ and $\mathbb B\models\neg\varphi$, it follows that $\mathbb A\models\neg\psi$ and $\mathbb B\models\psi$. Hence due to induction, Spoiler can win the game with starting node $\node BA$ for the formula $\psi$ with $\card\psi$ nodes, and thus wins the game for the formula $\varphi$ with $\card\varphi=\card\psi+1$ nodes.
 
 If $\varphi=\psi\vee\chi$, then Spoiler chooses sets $\mathbb A_1$ and $\mathbb A_2$ with $\mathbb A_1\cup\mathbb A_2=\mathbb A$ and $\mathbb A_1\models\psi$ and $\mathbb A_2\models\chi$. Clearly, $\mathbb B\models\neg\psi$ and $\mathbb B\models\neg\chi$. Therefore, by induction Spoiler can win the game for $\psi$ on $\node{A_1}B$ in $\card\psi$ moves, and can win the game for $\chi$ on $\node{A_2}B$ in $\card\chi$ moves. Therefore, Spoiler can win the game for $\varphi$ on $\node AB$ with $\card\psi+\card\chi+1=\card\varphi$ nodes as required.
 
 If $\varphi=\paramBox{\ssf}\psi$, then spoiler plays an $\ssf$ move as follows: Since $\mathbb B\models\neg\paramBox{\ssf}\psi$, Spoiler can choose a set $\mathbb B_1$ such that for each $(M,w)\in\mathbb B$ there is some $(M,w')$ with $w'\in\ssf(M,w)$ such that $M,w'\models\neg\psi$; it then follows that $\mathbb B_1\models\neg\psi$. On the other hand, since $\mathbb A\models\paramBox{\ssf}\psi$, for the set $\mathbb A_1=\set{(M,w')\ \vert\ (M,w)\in\mathbb A, w'\in\ssf(M,w)}$, we have that $\mathbb A_1\models\psi$. By induction, we therefore know that Spoiler can win the game on the mode $\node{A_1}{B_1}$ in $\card\psi$ moves, and hence can win the game on $\node AB$ for $\varphi$ in $\card\psi+1=\card\varphi$ moves as required.
 
 For the converse, assume that Spoiler can win the FSG($\setSsf$) starting with node $\node AB$ in $k$ moves; let $T$ be a corresponding tree with size $k$. Clearly, when we only consider the labels $p$, $\neg$, $\vee$ and $\ssf_i$, the tree $T$ represents a formula $\varphi$ from $\lmplus{\setSsf}$ with $\card\varphi=k$. By induction, we prove that for each node $v$ labelled with $\node AB$ in $T$, for the formula $\varphi_v$ represented by the subtree corresponding to $v$, we have that $\mathbb A\models\varphi_v$ and $\mathbb B\models\neg\varphi_v$.
 
 If $v$ is a leaf, then $v$ is labelled with a propositional variable $p$. Due to the winning condition, we know that $v$ is closed, hence $\mathbb A\models p$ and $\mathbb B\models\neg p$. Now assume that $v$ is not a leaf, then $v$ is labelled with $\neg$, $\vee$, or some $\ssf_i$ for $i\in\set{1,\dots,k}$.
 
 First assume that $v$ is labelled with $\neg$. Then $v$ has a single successor node $u$ labelled with $\node BA$, and $\varphi_v=\neg\varphi_{u}$. By induction, we know that $\mathbb B\models\varphi_{u}$, and $\mathbb A\models\neg\varphi_{u}$. Hence $\mathbb A\models\varphi_v$ and $\mathbb B\models\neg\varphi_v$ as required.
 
 Now assume that $v$ is labelled with $\vee$, then $v$ has two successor nodes $u_1$ and $u_2$ with $\varphi_v=\varphi_{u_1}\vee\varphi_{u_2}$ labelled with $\node{A_1}B$ and $\node{A_2}B$ with $\mathbb A_1\cup\mathbb A_2=\mathbb A$. By induction, we know that $\mathbb A_1\models\varphi_{u_1}$, $\mathbb A_2\models\varphi_{u_2}$, $\mathbb B\models\neg\varphi_{u_1}$ and $\mathbb B\models\neg\varphi_{u_2}$. Therefore, each pointed model $(M,w)\in\mathbb A$ satisfies $\varphi_{u_1}$ or $\varphi_{u_2}$, it follows that $\mathbb A\models\varphi_{u_1}\vee\varphi_{u_2}=\varphi$, and each pointed model $(M,w)\in\mathbb B$ satisfies $\neg\varphi_{u_1}$ and $\neg\varphi_{u_2}$, hence $\mathbb B\models\neg(\varphi_{u_1}\vee\varphi_{u_2})=\varphi$ as required.
 
 Finally assume that $v$ is labelled with $\ssf$. Then $v$ has a unique successor $u$, and $\varphi_v=\ssf\varphi_u$ for some $i\in\set{1,\dots,k}$, and $u$ is labelled with $\node{A_1}{B_1}$, where $\mathbb A_1=\set{(M,w')\ \vert\ (M,w)\in\mathbb A, w'\in\ssf_i(M,w)}$, and for each $(M,w)\in\mathbb B$, there is a pointed model $(M,w')\in\mathbb B_1$ with $w'\in\ssf(M,w)$. Due to induction, we know that $\mathbb A_1\models\varphi_u$, and $\mathbb B_1\models\neg\varphi_u$. By the choice of $\mathbb A_1$ and $\mathbb B_1$, it therefore follows that $\mathbb A\models\ssf\varphi_u=\varphi_v$, and $\mathbb B\models\neg\ssf\varphi_u=\neg\varphi_v$ as required.
\end{proof}

\subsubsection{Proof of Theorem~\ref{theorem:pigeonhole} (pigeonhole principle)}

\theorempigeonhole*

\begin{proof}
 Assume that this is not the case, and let there be a formula $\psi$ that is equivalent to $\varphi$. In particular, then $\mathbb A\models\psi$ and $\mathbb B\models\neg\psi$. Hence let $T^\psi_(\node AB)$ be the closed game tree $T\in\mathcal T(\node AB)$ that corresponds to playing the strategy $\psi$. Clearly, $T^\psi_(\node AB)$ is isomorphis to $\psi$.
 
 Then, by the proof of Theorem~\ref{theorem:fsg theorem}, each node of $\psi$ corresponds to a lead in $T^\psi_(\node AB)$, and each pointed model from $\mathbb A$ ($\mathbb B$) appeads in at least one leaf of $T^\psi_(\node AB)$. Since there are only $c$ pointed models in each leaf of $T^\psi_(\node AB)$, the tree $T^\psi_(\node AB)$ has at least $\frac{\card{\mathbb A}}c$ ($\frac{\card{\mathbb B}}c$). Therefore, $\varphi$ has at least this many leafs as well, and in particular, the size of $\varphi$ is at least $\frac{\card{\mathbb A}}c$ ($\frac{\card{\mathbb B}}c$) as claimed.
\end{proof}

\subsection{Proofs of Results in Section~\ref{proofs:single step operators} (Single Step Operators)}

\subsubsection{Proof of Theorem~\ref{theorem:1 step expressiveness -- or equivalence}}\label{sect:proof theorem:1 step expressiveness -- or equivalence}

\theoremonestepexpressivenessorequivalence*

\begin{proof}
 Since there are only a finite number of Boolean functions of each arity, let $\mathcal F=\set{f_1,\dots,f_k}$, and let $\mathcal G=\set{g_1,\dots,g_l}$. The direction \ref{theorem:1 step expressiveness -- or equivalence:or} to \ref{theorem:1 step expressiveness -- or equivalence:expressiveness} is trivial: If $g\equiv f_{i_1}\vee\dots\vee f_{i_t}$, then $M,w\models\paramBox{g}\varphi$ if and only if $M,w'\models\varphi$ for all worlds $w'$ such that $w'$ is an $I$-successor of $w$ for some $I$ with $f_{i_m}(I)=1$ for one of the $i_m$. Therefore, $\paramBox{g}\varphi$ is equivalent to $\paramBox{f_{i_1}}\varphi\wedge\dots\wedge\paramBox{f_{i_t}}\varphi$. 
 
 It remains to show that if one of the $g_i$s is not of this form, then there is a formula $\varphi$ of $\lmG$ that cannot be expressed in $\lmF$. A standard technique to prove such results are bisimulations, which we adapt to logics of this form (see details in Appendix~\ref{sect:appendix:bisimulations}).

 Hence assume indirectly that $\paramBox{g}p$ can be expressed in $\lmF$ and $g$ is not of the form $\bigvee_{i\in I}f_i$ for any set $I\subseteq\set{1,\dots,k}$. Let $S\subseteq\set{1,\dots,k}$ be the (possibly empty) set of indices $i$ such that $f_i$ implies $g$ (i.e., if $f_i(I)=1$, then $g(I)=1$). By choice of $S$, it follows that for each $i\notin S$, there is some assignment $I_i$ with $f_i(I_i)=1$ and $g(I_i)=0$. It also follows that $\bigvee_{i\in S}f_i$ implies $g$. Since we assumed that $g$ is not of the form $\bigvee_{i\in S}f_i$ for any $S$, it then follows that $g$ does not imply $\bigvee_{i\in S}f_i$. Therefore, there is an assignment $I_g$ such that $g(I_g)=1$, and for each $i\in S$, we have that $f_i(I_g)=0$. 

Now consider the following models $M_1$ on the left-hand side and $M_2$ on the right-hand side:
 
 \scalebox{0.5}{
\begin{tikzpicture}
 \node[draw,circle] at (0,0) (w1) {$w_1$};
 
 \node[draw,circle] at (-3,-3) (I1pos) {\phantom{$w_1$}};
 \node              at (-3,-3.75)       {$p$};
 
 \node[draw,circle] at (-2,-3) (I1neg) {\phantom{$w_1$}};
 \node              at (-2,-3.75)       {$\overline p$};
 
 \node              at (0,-3) {$\ldots\ldots$}; 

 \node[draw,circle] at (2,-3) (Ikpos) {\phantom{$w_1$}};
 \node              at (2,-3.75)       {$p$};
 
 \node[draw,circle] at (3,-3) (Ikneg) {\phantom{$w_1$}};
 \node              at (3,-3.75)       {$\overline p$};
 
 \node[draw,circle] at (5.5,-3) (Ig)    {$w^g_1$};
 \node              at (5.5,-3.75)       {$p$};

 \draw[->] (w1) edge node [auto=right] {$I_1$} (I1pos);
 \draw[->] (w1) edge node [auto=left] {$I_1$} (I1neg);

 \draw[->] (w1) edge node [auto=right] {$I_k$} (Ikpos);
 \draw[->] (w1) edge node [auto=left] {$I_k$} (Ikneg);
 
 \draw[->] (w1) edge node [auto=left] {$I_g$} (Ig);
\end{tikzpicture}}\ \ \ 
\scalebox{0.5}{
\begin{tikzpicture}
 \node[draw,circle] at (0,0) (w1) {$w_2$};
 
 \node[draw,circle] at (-3,-3) (I1pos) {\phantom{$w_1$}};
 \node              at (-3,-3.75)       {$p$};
 
 \node[draw,circle] at (-2,-3) (I1neg) {\phantom{$w_1$}};
 \node              at (-2,-3.75)       {$\overline p$};
 
 \node              at (0,-3) {$\ldots\ldots$}; 

 \node[draw,circle] at (2,-3) (Ikpos) {\phantom{$w_1$}};
 \node              at (2,-3.75)       {$p$};
 
 \node[draw,circle] at (3,-3) (Ikneg) {\phantom{$w_1$}};
 \node              at (3,-3.75)       {$\overline p$};
 
 \node[draw,circle] at (5.5,-3) (Ig)    {$w^g_2$};
 \node              at (5.5,-3.75)       {$\overline p$};

 \draw[->] (w1) edge node [auto=right] {$I_1$} (I1pos);
 \draw[->] (w1) edge node [auto=left] {$I_1$} (I1neg);

 \draw[->] (w1) edge node [auto=right] {$I_k$} (Ikpos);
 \draw[->] (w1) edge node [auto=left] {$I_k$} (Ikneg);
 
 \draw[->] (w1) edge node [auto=left] {$I_g$} (Ig);
\end{tikzpicture}}

Here, an arrow labelled $I_i$ between worlds $w$ and $w'$ represents that $(w,w')\in R_j$ for exactly those $j$ with $I_i(r_j)=1$. It is obvious that $M_1,w_1\models\paramBox{g}p$: Since $g(I_i)=0$ for all $i\in\set{1,\dots,k}$, the world $w^g_1$ is the only $g$-successor of $w_1$ in $M_1$, and by definition, $M_1,w^g_1\models p$. On the other hand, $M_2,w_2\models\neg\paramBox{g}p$, since the only $g$-successor of $w_2$ in $M_2$ is $w^g_2$, and by definition, $M_2,w^g_2\models\overline p$.

We define the relation $Z$ as follows: $Z$ contains the pair $(w_1,w_2)$ and all pairs of unnamed worlds in which $p$ has the same value. We show that $Z$ is an $\set{f_1,\dots,f_k}$-bisimulation. For the forward condition, let $w_1'$ be an $f_i$-successor of $w_1$. We distinguish two cases:

\begin{itemize}
 \item if $w_1'$ is \emph{not} the world $w^g_1$, then we can simply choose $w_2'$ to be the corresponding world $w_2'$ (i.e., the one in the same position in the picture) of model $M_2$, which is then an $f_i$-successor of $w_2$ with $(w_1',w_2')\in Z$.
 \item if $w_1'$ is the world $w^g_1$, then in particular, $w^g_1$ is an $f_i$-successor of $w_1$. It follows that $f_i(I_g)=1$, and therefore, $i\notin S$. By the choice of $I_i$, it follows that $f(I_i)=1$. Therefore, we can choose $w_2'$ as the $I_i$-successor $w_2$ in $M_2$ where $p$ is false.
\end{itemize}

The backward condition is shown analogously.

Now indirectly assume that there is a formula $\varphi$ of $\lmF$ which is equivalent to $\paramBox{g}p$. Then in particular it follows that $M_1,w_1\models\varphi$ and $M_2,w_2\not\models\varphi$. However, since the above-constructed bisimulation $Z$ contains the pair $(w_1,w_2)$ and $\varphi$ is a formula from $\lmplus{\set{f_1,\dots,f_n}}$, it follows from Proposition~\ref{prop:invariance under f1 ... fk bisimulation} that $M_1,w_1\models\varphi$ if and only if $M_2,w_2\models\varphi$. Hence we have a contradiction.
\end{proof}

\subsubsection{Proof of Theorem~\ref{theorem:finite step succinctness}}\label{appendix:proof of theorem:finite step succinctness}

\theoremfinitestepsuccinctness*

\begin{proof}
We prove the slightly stronger result that as soon as $\mathcal G$ contains a function $g$ which is a disjunction of functions in $\mathcal F$, but not an element of $\mathcal F$, then the formula $\neg\paramBox{g}^i\neg p$ needs exponential length when expressed as an $\lmF$-formula (an equivalent $\lmF$-formula does exist due to Theorem~\ref{theorem:1 step expressiveness -- or equivalence}).
 For $n$-ary Boolean functions $f_1$ and $f_2$, we write $f_1\leq f_2$ if $f_1(r_1,\dots,r_n)\leq f_2(r_1,\dots,r_n)$ for all $r_1,\dots,r_n\in\set{0,1}$. We define
 \begin{itemize}
  \item $\mathcal F_1=\set{f\in\mathcal F\ \vert\ f\leq g}$, and
  \item $\mathcal F_2=\mathcal F\setminus\mathcal F_1$.
 \end{itemize}
 
 Since $g$ is a disjunction of functions in $\mathcal F$, it then clearly follows that $g=\vee_{f\in\mathcal F_1}f$. Let $\mathcal F_1=\set{f_1,\dots,f_k}$, and let $\mathcal F_2={f_{k+1},\dots,f_m}$. For each $i\in\set{k+1,\dots,m}$, let $\overrightarrow{\alpha_i}=(r^i_1,\dots,r^i_n)$ be chosen such that $f_i(\overrightarrow{\alpha_i})=1$, and $g(\overrightarrow{\alpha_i})=0$. Such a sequence exists since $f_i\not\leq g$ for $i\in\set{k+1,\dots,m}$.
 
 Further, let a set of vectors $\set{\overrightarrow{\beta_1},\dots\overrightarrow{\beta_t}}$ be a smallest set chosen such that
 
 \begin{itemize}
  \item $g(\overrightarrow{\beta_1})=\dots=g(\overrightarrow{\beta_t})=1$,
  \item there is no $f\in\mathcal F_1$ with $f(\overrightarrow{\beta_1})=\dots=f(\overrightarrow{\beta_t})=1$.
 \end{itemize}
 
 Such a set exists, since for each $f\in\mathcal F_1$ we have that $f\leq g$ and $f\neq g$ (since $g\notin\mathcal F$). In particular, choosing $\overrightarrow{\beta_1},\dots\overrightarrow{\beta_t}$ as the set of all assignments $\overrightarrow\beta$ with $g(\overrightarrow\beta)=1$ satisfies the two conditions (although not minimality). Clearly, for the smallest such set, we still have that $t\ge2$, since for each $\overrightarrow\beta$ with $g(\overrightarrow\beta)=1$ there is some $f\in\mathcal F_1$ with $f(\overrightarrow\beta)=1$, since we know that $g=\vee_{f\in\mathcal F_1}f$.
 
 For each $i\in\mathbb N$, let $\varphi_i=\neg\paramBox{g}^i\neg p$. Clearly, $\varphi_i$ is a $\lmG$-formula (in fact even a $\lmplus{\set g}$-formula) and the length of $\varphi_i$ is linear in $i$. Since $g$ is a disjunction of functions in $\mathcal F$, due to Theorem~\ref{theorem:1 step expressiveness -- or equivalence}, it follows that for each $i$, there is some $\lmF$-formula $\psi_i$ of minimal length such that $\varphi_i$ and $\psi_i$ are equivalent. To prove the theorem, it suffices to show that the length of each $\psi_i$ is at least $(\frac{t}{t-1})^i$.
 
 To show this, we construct models similarly to the ones from the proof in~\cite{FrenchVanDerHoekIlievKooi-SUCCINCTNESS-MODAL-LOGIC-AI-2013}. Our models are based on trees of width $t$ and depth $i$, and are constructed as follows:
 
 \begin{itemize}
  \item Each tree $T$ has a root $w_0$ with depth $0$.
  \item Each node $u\in T$ with depth smaller than $i$ has successors $v^u_1,\dots,v^u_t$, where $v^u_i$ is a $\overrightarrow{\beta_i}$-successor of $u$. (Note that due to the minimality of $\overrightarrow{\beta_1},\dots\overrightarrow{\beta_t}$, the sequence consists of pairwise different vectors).
 \end{itemize}
 
 In this proof only, for a word $s=s_1\dots s_i\in\set{1,\dots,t}^*$, we say that a node $u\in T$ is an $s$-successor of a world $v$ if $s=\epsilon$ and $u=v$, or if there is an intermediate node $u'$ such that $u'$ is a $\overrightarrow{\beta_{s_1}}$-successor of $v$ and $u$ is (inductively) a $s_2\dots s_i$-successor of $u'$.
 
 We now define our models as follows:
 
 \begin{itemize}
  \item For each $s\in\set{1,\dots,t}^i$, let $A_s$ be the model obtained from the tree $T$, where in the unique world $w_s$ that is an $s$-successor of the root of $T$, the variable $p$ is true.
  \item The model $B$ is the model obtained from the tree $T$, where the variable $p$ is false in every world.
 \end{itemize}
 
 Additionally, if $u$ and $v$ are nodes with $\depth{v}=\depth{u}+1$ and one of the following is true:
 
 \begin{itemize}
  \item $u$ is a node of some $A_s$ and $v$ is a node of $B$, or
  \item $v$ is a node of some $A_s$ and $u$ is a node of $B$,
 \end{itemize}
 
 then $v$ is an $\alpha_j$-successor of $u$ for each $j\in\set{k+1,\dots,m}$.
 
 Let $\mathbb A=\set{A_s\ \vert\ s\in\set{1,\dots,t}^i}$, and $\mathbb B=\set{B}$. Then:
 
 \begin{itemize}
  \item $\mathbb A\models\varphi_i$, since the world $w_s$ satisfies the variable $p$, and $g(\overrightarrow{\beta_j})=1$ for all relevant $j$,
  \item $\mathbb B\models\neg\varphi_i$, since the model $B$ does not contain any world in which $p$ is true and which can be reached on a path adressed by $\paramBox{g}^i$.
 \end{itemize}
 
 We first show that no $T\in\mathcal T(\node AB)$ can contain a nontrivial node (i.e., a node labelled $\node CD$ with $\emptyset\notin\set{\mathbb C,\mathbb D}$) that is labelled with $\paramBox{f_j}$ for $j\in\set{k+1,\dots,m}$ (i.e., $f_j\in\mathcal F_2$). Recall from above that in this case, $g(\overrightarrow{\alpha_j})=0$, and $f_j(\overrightarrow{\alpha_j})=1$.
 Assume indirectly that such a node $u$ labelled with $\node CD$ exists. Let the successor node of $u$ be labelled with $\node{C_1}{D_1}$. 

 We make a case distinction:
 
 \begin{itemize}
  \item First assume that $\mathbb C$ corresponds to $\mathbb A$ and $\mathbb D$ corresponds to $\mathbb B$. Then $\mathbb C_1$ contains, in particular, all $\overrightarrow{\alpha_j}$-successors of all nodes in $\mathbb C$, which includes all nodes in $B$ of the corresponding depth. In particular, this includes the successor picked for the right-hand side in the model $B$. Therefore, we have a contradiction to Proposition~\ref{prop:closed game tree disjoint sets}.
  \item The second case is symmetric.
 \end{itemize}
 
 We theorefore know that the formulas $\psi_i$ do not contain any occurrance of an operator $\paramBox{f}$ for $f\in\mathcal F_2$, hence $\psi_i$ is in fact a $\lmplus{\mathcal F_1}$-formula. Hence to conclude the proof, it suffices to show that every leaf in a tree $T\in\mathcal T(\node AB)$ contains at most $(t-1)^i$ elements, the result then follows from Theorem~\ref{theorem:pigeonhole}, since $\card{\mathbb A}=t^i$. To show this, let $T\in\mathcal T(\node AB)$, and let $u$ be a leaf in $T$. Then $\boxlabels{u}=f_{j_1}\dots f_{j_i}$ (clearly, the modal depth of $\psi_i$ must be $i$), with $j_1,\dots,j_i\in\set{1,\dots,k}$. By construction for each $f_{j_l}$, there is one value $\overrightarrow{\beta_h}$ among $\overrightarrow{\beta_1},\dots,\overrightarrow{\beta_t}$ with $f_{j_l}(\overrightarrow{\beta_h})=0$. Therefore, for each of the $t$ successors of each node in each $A_s$, the application of $f_{j_l}$ covers at most $t-1$ many of them. Since the depth of the formula (and the tree $T$) is $i$, this implies that each leaf contains only at most $(t-1)^i$ many of the models $A_s$. 

\end{proof}

\subsubsection{Proof of Main Result of Single-Step Operators, Theorem~\ref{theorem:single step main result}}\label{section:proof of theorem:single step main result}

We now use the characterizations of expressiveness and succinctness obtained in Theorems~\ref{theorem:1 step expressiveness -- or equivalence} and~\ref{theorem:finite step succinctness} to prove our main result on single-step operators, Theorem~\ref{theorem:single step main result}.

\theoremsinglestepmainresult*

\begin{proof}
  Let $S=\set{s_1,\dots,s_n}$, with an ordering chosen such that if $s_i\leq_S s_j$, then $i\leq j$. Let $k=\lceil\log_2(\card S+1)\rceil$, then $2^k>\card S$. Hence there is an injective function $i\colon S\rightarrow\mathcal P(\set{1,\dots,k})$, such that $i(s)\neq\emptyset$ for all $s\in S$.
  
  We first prove the succinctness result, i.e., define the sets $\mathcal F_s$ for $s\in S$. For this, we use the $k$ modalities $\Box_1,\dots,\Box_k$. Let $\mathcal P$ contain all projections, i.e., all $k$-ary Boolean functions of the form $p_i(r_1,\dots,r_k)=r_i$ for some $i\in\set{1,\dots,k}$. We now define, for each $s\in S$, the function $f_s=\vee_{j\in i(s)}r_j$, and then define $\mathcal F_{s_i}$ inductively (recall that if $s_j\leq_S s_i$, then $j\leq i$) as
  $$\displaystyle\mathcal F_{s_i}=\mathcal P\cup\set{f_{s_i}}\cup\bigcup_{s_j\leq_S s_i}\mathcal F_{s_i}.$$
 
  Since all involved Boolean functions are disjunctions of functions in $\mathcal P$, and each $\mathcal F_{s_j}$ contains $\mathcal P$ as a subset, it follows from Theorem~\ref{theorem:1 step expressiveness -- or equivalence} that all $\lmplus{\mathcal F_s}$ are equally expressive as $\lmplus{\mathcal P}$. In particular, all $\lmplus{\mathcal F_s}$ are equally expressive. By construction, if $s\leq_S t$, then $\mathcal F_s\subseteq\mathcal F_t$, and hence in particular, every $\lmplus{\mathcal F_s}$-formula is also a $\lmplus{\mathcal F_t}$-formula as claimed.
 
  Now assume that $s\not\leq_S t$, and let $s=s_i$, $t=s_j$ for some $i,j\in\set{1,\dots,n}$. By construction, it follows that the function $f_s$ is an element of $\mathcal F_{s_i}$, but not an element of $\mathcal S_{s_j}$. Since $\lmplus{\mathcal F_{s_i}}$ and $\lmplus{\mathcal F_{s_j}}$ are equally expressive, Theorem~\ref{theorem:finite step succinctness} then implies that $\lmplus{\mathcal F_{s_i}}$ is exponentially more succinct than $\lmplus{\mathcal F_{s_j}}$. This completes the proof.
 
  For the expressiveness result, we use a very similar construction, but leave out the projections (as their role was to ensure that all logics have the same expressive power).  We define the function $g_s=\oplus_{j\in i(s)}r_j$, and define the sets $\mathcal G_s$ as follows (inductively as above):
  $$\displaystyle\mathcal G_{s_i}=\set{g_s}\cup\bigcup_{s_j\leq_S s_i}\mathcal G_{s_i}.$$
  The proof is identical to the succinctness case above, since Theorem~\ref{theorem:1 step expressiveness -- or equivalence} implies that $\paramBox{f_s}$ cannot be expressed with any number of opeators $\paramBox{f_{s'}}$ for $s'\neq s$.
\end{proof}

\subsection{Proofs of Results in Section~\ref{proofs:arbitrary step operators} (Arbitrary-Step Operators)}\label{appendix:arbitrary step proofs}

\subsubsection{Modal Depth}

For a node $v\in T$ where $\boxlabels{v}=L_1\dots L_m$ for length-uniform languages $L_1,\dots,L_m$, we say that the \emph{modal depth} of $v$ is the value $\dcard{L_1\circ\dots\circ L_m}$. We denote this value with $\md v$. A straight-forward induction on the path from the root to the node $v$ shows the following:

\begin{proposition}\label{prop:depth of covered worlds}
 Let $v\in T\in\mathcal T_\setSsf(\modelClassA,\modelClassB)$, where $\setSsf$ contains only length-uniform languages. Let $(\mathbb X,w_X)$ be covered by $v$, where $\mathbf X\in\set{\mathbf A,\mathbf B}$. Then $\depth{w_X}=\md v$.
\end{proposition}

We say that a model $M$ is \emph{complete}, if every world $w\in M$ has both a $1$- and a $2$-successor. Note that all models $\starA$ and $\starB$ are complete, but $\baseA$ and $\baseB$ are not. We say that a node $v$ \emph{covers} a pair of models $(M_A,M_B)$ if $v$ is labelled $\node CD$ and there are worlds $w_A\in M_A$ and $w_B\in M_B$ such that $(M_A,w_A)\in\mathbb C$ and $(M_B,w_B)\in\mathbb D$ or $(M_A,w_A)\in\mathbb D$ and $(M_B,w_B)\in\mathbb C$. In particular, then $v$ is a subformula of $\varphi$ that distinguishes $(M_A,w_A)$ and $(M_B,w_B)$ in the sense that $M_A,w_a\models\varphi$ if and only if $M_B,w_b\models\neg\varphi$.

\subsubsection{Proof of Lemma~\ref{lemma:minimal formula gives trap avoiding tree}}

\lemmaminimalformulagivestrapavoidingtree*

\begin{proof}
 Since $\psi$ is equivalent to $\paramBox{A_\ell}^ip$, we know that $\psi$ is $p$-monotone in the following sense: If $M$ and $M'$ are models where $M'$ is obtained from $M$ by making $p$ true in additional worlds and $M,w\models\psi$, then $M',w\models\psi$ holds as well. Since $\psi$ is minimal, this implies that $\psi$ only contains positive occurrences of $p$, i.e., the variable $p$ only occurs under an even numer of negations. Furthermode, $p$ is the only propositional variable appearing in $\psi$. Therefore, every leaf of $\parseTree$ is labelled with the variable $p$ and a class of models $\node CD$ such that $\mathbb C$ corresponds to $\modelClassA$, and $\mathbb D$ to $\modelClassB$. 
 
 We now show inductively that for every node $v\in\parseTree$ labelled with $\node CD$ or $\node DC$ such that $\mathbb C$ corresponds to $\modelClassA$ and $\mathbb D$ corresponds to $\modelClassB$, for every model $(\starA,w_A)\in\mathbb C$ there is a descendent of $w_A$ in $\starA$ where $p$ is true, and for every model $(\starB,w_B)\in\mathbb D$, there is a descendent of $w_B$ in $\starB$ where $p$ is false. Since the reflexive singleton of $\starA$ ($\starB$) does not have a successor where $p$ is true (false), this shows that $\parseTree$ does not contain a node that is labelled with a model $(\starA,\wtrap)$ or $(\starB,\wtrap)$.
 
 We prove this claim inductively over the tree structure. For the leaves, the claim follows from the above, as every leaf is labelled with $p$ and $\node CD$ where $\mathbb C$ ($\mathbb D$) corresponds to $\modelClassA$ ($\modelClassB$), hence every pointed models in $\mathbb C$ satisfy $p$, and all pointed models in $\mathbb B$ satisfy $\overline p$.
 
 Now let $v$ be a non-leaf node labelled $\node CD$ in $\parseTree$ such that $v$ is not a leaf and the above claim is true for all successors of $v$. We make a case distinction:
 
 \begin{itemize}
  \item If $v$ is labelelled $\neg$, then $v$ has a single successor $v'$ for which the claim holds by induction. The result for $v$ follows trivially, since $v'$ is labelled $\node DC$ or $\node CD$, where $\mathbb D$ corresponds to $\modelClassB$ and $\mathbb C$ to $\modelClassA$.
  \item If $v$ is labelled $\vee$, then $v$ has two successors labelled $\node{C_1}{D}$ and $\node{C_2}{D}$ with $\mathbb C_1\cup\mathbb C_2=\mathbb D$, and for which the claim is true. Since the claim is true for each model in $\mathbb C_1$ and $\mathbb C_2$, it is also true for their union, $\mathbb C$.
  \item If $v$ is laballed $\paramBox{L}$, then $v$ has a single successor $v'$ labelled $\node{C_1}{D_1}$, where $\mathbb C_1$ ($\mathbb D_1$) contains (at least) one descendent for each pointed model in $\mathbb C$ ($\mathbb D$). Since the claim is true for the sets $\mathbb C_1$ and $\mathbb D_1$ and the descendent relation is transitive, the claim for $\mathbb C$ and $\mathbb D$ follows.
 \end{itemize}

 This completes the proof.
\end{proof}

\subsubsection{Proof of Lemma~\ref{lemma:alternation:splitting only at multilpes of l}}

\lemmaalternationsplittingonlyatmultiplesofl*

\begin{proof}
 Assume that this is not the case, and let $d:=\dcard{L_1\circ\dots\circ L_{m-1}}$, let $v$ be labelled $\node CD$. By Proposition~\ref{prop:depth of covered worlds}, for every pointed model $(M,w)\in\mathbb C\cup\mathbb D$, we have that $\depth{w}=d$. By construction, if $d$ is not a multiple of $\ell$, then each world $w\in\starA$ or $\starB$ with $\depth{w}=d$ does not have both a $1$- and a $2$-successor in the base model $\baseA$ ($\baseB$). Since $A_{\ell'}$ contains the word $\variableLAltword{\ell'}{1}$ (starting with $1$) and $\variableLAltword{\ell'}{2}$ (starting with $2$), the successor node $v'$ of $v$, labelled with $\node{C_1}{D_1}$, contains a model $(\starA,w_A)$or $(\starB,w_B)$ in $\mathbb C_1$ where $w_A\notin\baseA$ or $w_B\notin\baseB$, i.e., $w_A=\wtrap$ or $w_B=\wtrap$. This is a contradiction to Lemma~\ref{lemma:minimal formula gives trap avoiding tree}.
\end{proof}

\subsubsection{Proof of Lemma~\ref{lemma:alternation: a-l-s always in language}}

\lemmaalternationalsalwaysinlanguage*

\begin{proof}
Due to Lemma~\ref{lemma:minimal formula gives trap avoiding tree}, we know that for every covered model $(\starB,w_B)$, $w_B$ is a world of the base model $\baseB$. Due to Proposition~\ref{prop:depth of covered worlds}, we know that for $(\starB,w_B)$ as above, $\depth{w_B}=d$.
Therefore, $w_B$ is the unique $\altword s[1\dots d]$-successor of $\starB$'s root in $\starB$. Hence, due to Lemma~\ref{lemma:fsg path property}, we know that $\altword s[1\dots d]\in L$.
\end{proof}

\subsubsection{Proof of Lemma~\ref{lemma:l' not multiple of l forces constant}}

\lemmalprimenotmultipleoflforcesconstant*

\begin{proof}
 Since $L_i=A_{f\cdot\ell+q}$ with $q\ge1$, Lemma~\ref{lemma:alternation:splitting only at multilpes of l} implies that $\dcard{L_1\circ\dots\circ L_{i-1}}=g\cdot\ell$ for some $g$. Let $u=\frac1\ell\cdot\dcard{L_1\circ\dots\circ L_{i-1}}+f+1=g+f+1$ and let $d':= \dcard{L_1\circ\dots\circ L_i}$, then $d'=g\cdot\ell + f\cdot\ell+q'=(g+f)\cdot\ell+q'$. Since this is not a multiple of $\ell$, Lemma~\ref{lemma:alternation:splitting only at multilpes of l}, implies that $L_{i+1}$ (and $i+1$ does exist, since otherwise the formula does not have the full modal depth) cannot be $A_{\ell'}$ for any $\ell'$, hence $L_{i+1}=\set{\alpha}$ for some $\alpha\in\set{1,2}$. Then it follows that for each word $x\in L$, we have that $x[d'+1]=x[(g+f)\cdot\ell+q'+1]=\alpha$. Due to Lemma~\ref{lemma:alternation: a-l-s always in language}, we know that $\altword{s_1},\altword{s_2}\in L$. With the above, this implies that $\altword{s_1}[(g+f)\cdot\ell+q'+1]=\altword{s_2}[(g+f)\cdot\ell+q'+1]=\alpha$.
 
 Now indirectly assume that $s_1[u]\neq s_2[u]$, i.e., $s_1[g+f+1]\neq s_2[g+f+1]$. In particular, then $\altword{s_1}[(g+f)\cdot\ell+q'+1]\neq a^\ell_{s_2}[(g+f)\cdot\ell+q'+1]$, which is a contradiction to the above.
\end{proof}

\subsubsection{Proof of Lemma~\ref{lemma:l' > l forces identities}}

\lemmalprimegreaterthanlforcesidentities*

\begin{proof}
 Withous loss of generality, we assume $\ell``>\ell$, since otherwise, there is no $j$ in the required interval and the claim is trivial. With $d_i$, we again denote $\dcard{L_1\circ\dots\circ L_i}$. Since $L_i=A_{\ell'}$ with $\ell'\ge1$, we know from Lemma~\ref{lemma:alternation:splitting only at multilpes of l} that $d_{i-1}$ is a multiple of $\ell$, and by choice of $u$ it follows that $d_{i-1}=u\cdot\ell$. In particular, $u$ is a natural number.
 
 Let $L=L_1\circ\dots\circ L_m$. Then, due to Lemma~\ref{lemma:alternation: a-l-s always in language}, we know that $\altword s[1\dots d_m]\in L$. 
 
 Since $L_i=A_{\ell'}$, we know that for each word $x\in L$, the subword $x[d_{i-1}+1\dots d_{i-1}+\ell']$ is alternating.
 
 Therefore, since $\altword s\in L$ and $d_{i-1}=u\cdot\ell$, we know that $\altword s[u\cdot\ell+1\dots u\cdot\ell+\ell']$ is alternating, i.e., for each position $i\in\set{u\cdot\ell+1,\dots,u\cdot\ell+\ell'-1}$, we have that $\altword s[i]\neq\altword s[i+1]$. 
 
 Now let $1\leq j<\frac{\ell'}\ell1$, and let $i=(u+j)\cdot\ell$. Then, since $j\ge1$ and $\ell\ge1$, it follows that $i=(u+j)\cdot\ell\ge(u+1)\cdot\ell\geq u\cdot\ell+1$, and since $j<\frac{\ell'}\ell$, we have that $i=(u+j)\cdot\ell<(u+\frac{\ell'}\ell)\cdot\ell=u\cdot\ell+\ell'$, and hence $i\leq u\cdot\ell+\ell'-1$. Therefore, $i$ is in the above interval, and hence $\altword s[i]\neq\altword s[i+1]$.
 
 Since $i=(u+j)\cdot\ell$, by the definition of $\altword s$, it follows that $\altword s[i]$ is the last symbol of $\altword{s[u+j]}$, and $\altword{s[i+1]}$ is the first symbol of $\altword{s[u+j+1]}$.
 
 Hence we know that the last symbol of $\altword{s[u+j]}$ is different from the first symbol of $\altword{s[u+j+1]}$.
 
 \begin{itemize}
   \item If $\ell$ is even, then for both $\alpha\in\set{1,2}$, the first symbol of $\altword\alpha$ is $\alpha$, and the last symbol of $\altword\alpha$ is $3-\alpha$. Hence $\alpha=s[u+j]\neq s[u+j+1]=3-\alpha$ would imply that the last symbol of $\altword{s[u+j]}$ (namely $3-\alpha$) is identical to the first symbol of $\altword{s[u+j+1]}$ (which is also $3-\alpha$), but from the above we know that the last symbol of $\altword{s[u+j]}$ is different from the first symbol of $\altword{s[u+j+1]}$. Hence in this case $s[u+j]=s[u+j+1]$.
   \item If $\ell$ is odd, then for both $\alpha\in\set{1,2}$, both the first and the last symbol of $\altword\alpha$ is $\alpha$. Since the last symbol of $\altword{s[u+j]}$ is different from the first symbol of $\altword{s[u+j+1]}$, this implies that $s[u+j]\neq s[u+j+1]$. Since only the symbols $1$ and $2$ appear, this means that $s[u+j]=3-s[u+j+1]$.
 \end{itemize}
\end{proof}

\subsubsection{Proof of Theorem~\ref{theorem:alternation expressiveness}}

\theoremalternationexpressiveness*

\begin{proof}
 Assume that such a formula exists, and let $\psi$ be one of minimal size. Define classes of models $\mathbb A=\oneModelClassA$ and $\mathbb B=\oneModelClassB$. Then $\mathbb A\models\psi$ and $\mathbb B\models\neg\psi$. Let $T=\parseTree$. Then $T$ only contains the operators available in $\lmA I$. Clearly, there is a leaf $v$ of $T$ that covers the string $s=1$. Clearly, every ancestor of $v$ covers the string $s$ as well. Let $\boxlabels{v}=L_1\dots L_m$. Then for each $i$, we have that $L_i$ is of the form $A_{\ell'}$ for some $\ell'\neq l$. In particular, we have that $L_1=A_{\ell'}$ for some $\ell'\neq\ell$. Clearly, we can without loss of generality assume that $\ell'<\ell$. From Lemma~\ref{lemma:alternation:splitting only at multilpes of l}, it then follows that $L_2$ cannot be of the form $A_{\ell'}$ for any $\ell'\in\mathbb N$. However, since in the logic $\lmA{I}$ with $\ell\notin I$, only languages of this form occur, we have a contradiction.
\end{proof}

\subsubsection{Proof of Lemma~\ref{lemma:restriction on sets by one branch}}

\lemmarestrictiononsetsbyonebranch*

\begin{proof}
 Let $\boxlabels{v}=L_1\dots L_m$. We consider each $j$ for which $\dcard{L_1\circ\dots\circ L_{j-1}}$ is a multiple of $\ell$, and show how the operator $L_i$ restricts the possible values of $s$.
 Hence let $\dcard{L_1\circ\dots\circ L_{j-1}}=g\cdot\ell$, we say that $g\cdot\ell$ is the depth in which this operator appears. There are three cases to consider.
 
 \begin{enumerate}
  \item If $L_i=\set\alpha$ for some $\alpha\in\set{1,2}$, then $L$ contains only words $x$ with $x[g\cdot\ell+1]=\alpha$. Due to Lemma~\ref{lemma:alternation: a-l-s always in language}, we know that $\altword s\in L$ for all $s\in S$, and hence $\altword s[g\cdot\ell+1]=\alpha$, which implies that $s[g+1]=\alpha$. Hence this operator rules out $\frac 12$ of all possible strings in $\set{1,2}^i$.
  Due to Lemma~\ref{lemma:alternation:splitting only at multilpes of l}, the languages $L_{i+1},\dots, L_{i+\ell-2}$ are not of the form $A_{\ell'}$ for some $\ell'$, hence the next restriction occurs at depth $(g+1)\cdot\ell$.
  \item If $L_i=A_{f\cdot\ell}$ for some $f>1$ (recall that $L_i\neq A_\ell$ for all $i$), then, by Lemma~\ref{lemma:l' > l forces identities}, the elements of $S$ must satisfy a sequence of $(f-1)$ equalities. Hence this operator rules out all but $\frac1{2^{(f-1)}}$ strings in $\set{1,2}^i$, and, again due to Lemma~\ref{lemma:alternation:splitting only at multilpes of l}, the next restriction appears at depth $(g+f)\cdot\ell$.
  \item If $L_i=A_{f\cdot\ell+q}$ for some $f\ge 0$ and $1\leq q<\ell$, then, by Lemma~\ref{lemma:l' > l forces identities}, the elements of $S$ must satisfy a sequence of $f$ identities (one identity for each $j\in\set{0,\dots,f-1}$).
  Hence all but $\frac1{2^f}$ elements of $\set{1,2}^i$ are ruled out, and, as above, the next restriction appears at the next multiple of $\ell$, i.e., at depth $(g+f+1)\cdot\ell$.
 \end{enumerate}
 
 Note that is is easy to see that the conditions required by $L_i$ at different indices are independent, as they refer to different indices of the strings $s$. Hence the following three operations appear:
 
 \begin{itemize}
  \item Increase depth by $\ell$, and add a restriction factor of $\frac12$,
  \item Increase depth by $f\cdot\ell$, and add a restriction factor of $\frac1{2^{f-1}}$,
  \item Increase depth by $(f+1)\cdot\ell$, and add a restriction factor of $\frac1{2^f}$.
 \end{itemize}
 
 In each case, increasing the depth by $2\cdot\ell$ adds a restrictin factor of at least $\frac12$. Since the complete depth must be $i\cdot\ell$, this means that the minimum restriction factor is at most $\frac{1}{2^{\frac i2}}$, i.e., we have that $\card S\leq 2^i\cdot\frac{1}{2^{\frac i2}}=2^{\frac i2}$ as claimed.
\end{proof}

\subsubsection{Proof of Theorem~\ref{theorem:alternation succinctness}}

\theoremalternationsuccinctness*

\begin{proof}
 Recall that $\modelClassB=\set{\starB\ \vert\ s\in\set{1,2}^i}$. In particular $\card{\mathbb B}=2^i$, and that $\modelClassA\models\psi$ and $\modelClassB\models\neg\psi$. Due to Theorem~\ref{theorem:pigeonhole}, it thus suffices to show that each leaf $u$ of a tree $T\in\mathcal T(\angNode\modelClassA\modelClassB)$ covers at most $2^{\frac i2}$ strings $s\in\set{1,2}^i$. Clearly it is enough to show the result for $T=\parseTree$, since $\psi$ is a formula equivalent to $\paramBox{A_\ell}^i$ of minimal size. Without loss of generality we can assume that the modal depth of each leaf ov $T$ is exactly $i\cdot\ell$. Clearly, for each $s\in\set{1,2}^i$, there is a leaf $v_s$ of $T$ that covers $s$.  Lemma~\ref{lemma:restriction on sets by one branch} states that each leaf $v_s$ can cover at most $2^{\frac i2}$ elements as claimed. Therefore, $T$ must have at least $\frac{2^i}{2^{\frac i2}}=2^{\frac i2}$ leaves, which concludes the proof.
\end{proof}

\subsubsection{Proof of Main Result on Alternation Languages, Theorem~\ref{theorem:main alternation result: subset equivalence}}

\theoremmainalternationresultsubsetequivalence*

\begin{proof}
  Since $I_1\nsubseteq I_2$, there is some $\ell\in I_1\setminus I_2$.
  \begin{enumerate}
   \item For each $i\in\mathbb N$, define $\varphi_i=\paramBox{A_\ell}^{2i}p$. Then clearly, the length of $\varphi_i$ is linear in $i$. For each $i$, let $\psi_i$ be the smallest formula in $\lmA{I_2}$ that is equivalent to $\varphi_i$. Then, due to Theorem~\ref{theorem:alternation succinctness}, we know that $\card{\psi_i}\ge 2\cdot2^{\frac i2}$, hence the length of $\psi_i$ is exponential in the length of $\varphi_i$ as claimed.
   \item The formula $\paramBox{A_\ell}p$ is a $\lmA{I_1}$-formula, and due to Theorem~\ref{theorem:alternation expressiveness}, there is no $\lmA{I_2}$-formula equivalent to $\paramBox{A_\ell}p$.
  \end{enumerate}
\end{proof}

\subsubsection{Proof of Main Result on Arbitrary-step Operators, Theorem~\ref{theorem:arbitrary step main result}}

\theoremarbitrarystepmainresult*

\begin{proof}
 The proof is very similar to the proof of Theorem~\ref{theorem:single step main result}. Let $S=(s_\ell)_{\ell\in\mathbb N}$, where $s_j\leq_S s_i$ implies $j\leq i$. Now, for $\ell\ge1$, inductively define $I_{s_\ell}$ as follows:
 
 $$I_{s_\ell}=\set{\ell}\cup\bigcup_{t\leq_S s}I_t.$$
 
 Clearly, we have that $I_s\subseteq I_t$ if and only if $s\leq_S t$. Therefore, the result follows from Theorem~\ref{theorem:main alternation result: subset equivalence} with the choice $\mathcal L_s=\set{A_\ell\ \vert\ \ell\in I_s}\cup\set{\set1,\set2}$ and $\mathcal K_s=\set{A_\ell\ \vert\ \ell\in I_s}$.
\end{proof}

\end{appendix}


\begin{thebibliography}{FvdHIK13}

\bibitem[ABvdT10]{AucherBoellavdTorre-PRIVACY-WITH-MOCAL-LOGIC-DEON-2010}
Guillaume Aucher, Guido Boella, and Leendert van~der Torre.
\newblock Privacy policies with modal logic: The dynamic turn.
\newblock In Guido Governatori and Giovanni Sartor, editors, {\em DEON}, volume
  6181 of {\em Lecture Notes in Computer Science}, pages 196--213. Springer,
  2010.

\bibitem[AI03]{AdlerImmerman-LOWER-BOUND-FORMULA-SIZE-TOCL-2003}
Micah Adler and Neil Immerman.
\newblock An {\it n!} lower bound on formula size.
\newblock {\em ACM Trans. Comput. Log.}, 4(3):296--314, 2003.

\bibitem[BdRV01]{BlackburnDeRijkeVenama-MODAL-LOGIC-BOOK-2001}
Patrick Blackburn, Maarten de~Rijke, and Yde Venema.
\newblock {\em Modal Logic}, volume~53 of {\em Cambridge Tracts in Theoretical
  Computer Scie}.
\newblock Cambridge University Press, Cambridge, 2001.

\bibitem[FHMV95]{FaginHalpernMosesVardi-REASONING-ABOUT-KNOWLEDGE-MITPRESS-1995}
Ronald Fagin, Joseph~Y. Halpern, Yoram Moses, and Moshe~Y. Vardi.
\newblock {\em Reasoning about Knowledge}.
\newblock 1995.

\bibitem[FvdHIK13]{FrenchVanDerHoekIlievKooi-SUCCINCTNESS-MODAL-LOGIC-AI-2013}
Tim French, Wiebe van~der Hoek, Petar Iliev, and Barteld~P. Kooi.
\newblock On the succinctness of some modal logics.
\newblock {\em Artif. Intell.}, 197:56--85, 2013.

\bibitem[GKPS95]{GogicKautzPapadimitriouSelman-COMPARATIVE-LINGUISTICS-KNOWLEDGE-REPRESENTATION-IJCAI-1995}
Goran Gogic, Henry~A. Kautz, Christos~H. Papadimitriou, and Bart Selman.
\newblock The comparative linguistics of knowledge representation.
\newblock In {\em Proceedings of the Fourteenth International Joint Conference
  on Artificial Intelligence, {IJCAI} 95, Montr{\'{e}}al Qu{\'{e}}bec, Canada,
  August 20-25 1995, 2 Volumes}, pages 862--869. Morgan Kaufmann, 1995.

\bibitem[GPT87]{GargovPassyTinchev-MODAL-ENVIRONMENT-BOOLEAN-SPECULATIONS-MLIA-1987}
George Gargov, Solomon Passy, and Tinko Tinchev.
\newblock Modal environment for boolean speculations.
\newblock In {\em Mathematical logic and its applications}, pages 253--263.
  Springer, 1987.

\bibitem[HM92]{hamo92}
J.~Halpern and Y.~Moses.
\newblock A guide to completeness and complexity for modal logics of knowledge
  and belief.
\newblock {\em Artificial Intelligence}, 54(2):319--379, 1992.

\bibitem[vdHI14]{VanderhoekIliev-RELATIVE-SUCCINCTNESS-AAMAS-2014-ACM-ENTRY}
Wiebe van~der Hoek and Petar Iliev.
\newblock On the relative succinctness of modal logics with union, intersection
  and quantification.
\newblock In {\em AAMAS}, AAMAS '14, pages 341--348, Richland, SC, 2014.
  International Foundation for Autonomous Agents and Multiagent Systems.

\bibitem[Wil99]{Wilke-CTL-SUCCINCTNESS-FSTTCS-1999}
Thomas Wilke.
\newblock {CTL}$^{\mbox{+}}$ is {E}xponentially more {S}uccinct than {CTL}.
\newblock In C.~Pandu Rangan, Venkatesh Raman, and Ramaswamy Ramanujam,
  editors, {\em FSTTCS}, volume 1738 of {\em Lecture Notes in Computer
  Science}, pages 110--121. Springer, 1999.

\end{thebibliography}
\end{document}